\documentclass[a4paper,USenglish]{llncs}

\usepackage[inline]{enumitem}
\PassOptionsToPackage{hyphens}{url}\usepackage{hyperref}

\expandafter\def\expandafter\UrlBreaks\expandafter{\UrlBreaks
  \do\a\do\b\do\c\do\d\do\e\do\f\do\g\do\h\do\i\do\j%
  \do\k\do\l\do\m\do\n\do\o\do\p\do\q\do\r\do\s\do\t%
  \do\u\do\v\do\w\do\x\do\y\do\z\do\A\do\B\do\C\do\D%
  \do\E\do\F\do\G\do\H\do\I\do\J\do\K\do\L\do\M\do\N%
  \do\O\do\P\do\Q\do\R\do\S\do\T\do\U\do\V\do\W\do\X%
  \do\Y\do\Z}

\usepackage{amssymb}
\usepackage{amsmath}

\usepackage{stmaryrd}
\usepackage{xspace}

\usepackage{thmtools}
\usepackage{thm-restate}

\usepackage[obeyDraft]{todonotes}
\usepackage{bussproofs}

\newcommand{\R}{\mathcal{R}}
\newcommand{\U}{\mathcal{U}}
\renewcommand{\P}{\mathcal{P}}

\newcommand{\CF}{{\rm CF}}

\newcommand{\eqbydef}{\triangleq}
\newcommand{\tsem}[1]{[\![{#1}]\!]}
\newcommand{\sem}[1]{\lfloor\!\!\lfloor{#1}\rfloor\!\!\rfloor}
\newcommand{\gtran}[1]{\leadsto_{#1}}
\newcommand{\gpath}[1]{\circ_{#1}}
\newcommand{\rform}[2]{{#1}\Rightarrow{#2}}
\newcommand{\rrule}[3]{{#1}\twoheadrightarrow{#2}{\tt~if~}#3}

\newcommand{\lbrbrak}{\left<}
\newcommand{\rbrbrak}{\right>}

\newcommand{\ct}[2]{\ensuremath{\lbrbrak{#1}\,\middle|\,{#2}\rbrbrak}}

\newcommand{\var}{\mathit{var}}

\newcommand{\bltn}[1]{{#1}^{\sf b}}
\newcommand{\cons}[1]{{#1}^{\sf c}}

\newcommand{\hd}{{\it hd}}
\newcommand{\ls}{{\it ls}}
\newcommand{\dom}{{\it dom}}

\newcommand{\Bool}{{\it Bool}\xspace}
\newcommand{\Int}{{\it Int}\xspace}
\newcommand{\Cfg}{{\it Cfg}\xspace}

\newcommand{\limplies}{\rightarrow}
\renewcommand{\iff}{\leftrightarrow}
\newcommand{\True}{\top}
\newcommand{\False}{\bot}

\newcommand{\demEntails}{\vdash^{\forall}}

\newcommand{\demModels}{\vDash^{\forall}}

\renewcommand{\models}{\vDash}

\newcommand{\semrule}[1]{\llbracket {\sf {#1}} \rrbracket}

\newcommand{\SHSH}{\ensuremath{\textit{SH}}}

\newcommand{\qeq}{=}

\newcommand{\X}{\mathcal{X}}

\title{A Coinductive Approach to Proving Reachability Properties in
  Logically Constrained Term Rewriting Systems}

\author{{\c S}tefan Ciob\^ac\u{a}\inst{1} \and Dorel Lucanu\inst{1}}
\institute{
  Alexandru Ioan Cuza University, Romania\\
  \texttt{\{stefan.ciobaca,dlucanu\}@info.uaic.ro}
}

\begin{document}

\maketitle
\begin{abstract}

  We introduce a sound and complete coinductive proof system for
  reachability properties in transition systems generated by logically
  constrained term rewriting rules over an order-sorted signature
  modulo builtins.
  A key feature of the calculus is a circularity proof rule, which
  allows to obtain finite representations of the infinite coinductive
  proofs.

\end{abstract}

\section{Introduction}
\label{sec:introduction}

We propose a framework for specifying and proving reachability
properties of systems whose behaviour is modelled using transition
systems described by logically constrained term rewriting systems
(LCTRSs). By reachability properties we mean that a set of target
states are reached in all terminating system computations starting
from a given set of initial states. We assume transition systems are
generated by constrained term rewriting rules of the form
\\ \centerline{$\rrule{l}{r}{\phi},$} \\ where $l$ and $r$ are terms
and $\phi$ is a logical constraint. The terms $l, r$ may contain both
uninterpreted function symbols and function symbols interpreted in a
builtin model, e.g., the model of booleans and integers. The
constraint $\phi$ is a first-order formula that limits the application
of the rule and which may contain predicate symbols interpreted in the
builtin model. The intuitive meaning of a constrained rule
$\rrule{l}{r}{\phi}$ is that any instance of $l$ that satisfies $\phi$
transitions in one step into a corresponding instance of $r$.

\begin{example}\label{ex:algorithm_compositeness}
  The following set of constrained rewrite rules specifies a procedure
  for compositeness:
\\
\centerline{$\begin{array}{l}
\rrule{\textit{init}(n)}{\textit{loop}(n, 2)}{\True}, \\
\rrule{\textit{loop}(i \times k, i)}{\textit{comp}}{k > 1}, \\
\rrule{\textit{loop}(n, i)}{\textit{loop}(n, i + 1)}{\lnot(\exists k.k
      > 1 \land n = i \times k)}.
\end{array}$}
\\
If $n$ is not composite, the computation of the procedure is infinite.
\end{example}
Given a LCTRS, which serves as a specification for a transition
system, it is natural to define the notion of constrained term
$\ct{t}{\phi}$, where $t$ is an ordinary term (with variables) and
$\phi$ is a logical constraint. The intuitive meaning of such a term
is the set of ground instances of $t$ that satisfy $\phi$.
\begin{example}The constrained term
  $\ct{\textit{init}(n)}{\exists u.1 < u < n \land n
    \mathrel{\textit{mod}} u = 0}$ defines exactly the instances of
  $\textit{init}(n)$ where $n$ is composite.
\end{example}

A \emph{reachability formula} is a pair of constrained terms
$\rform{\ct{t}{\phi}}{\ct{t'}{\phi'}}.$ The intuitive meaning of a
reachability formula is that any instance of $\ct{t}{\phi}$ reaches,
along all terminating paths of the transition system, an instance of
$\ct{t'}{\phi'}$ that agrees with $\ct{t}{\phi}$ on the set of shared
variables.

\begin{example}\label{ex:reachability-formula}
  The reachability formula\\
  \centerline{$\rform{\ct{\textit{init}(n)}{\exists u.1 < u <
        n \land n \mathrel{\textit{mod}} u = 0}}{\ct{\textit{comp}}{\True}}$}\\
  captures a functional specification for the algorithm described in
  Example~\ref{ex:algorithm_compositeness}: each terminating
  computation starting from a state in which $n$ is composite reaches
  the state \textit{comp}. Computations that start with a negative
  number (composite or not) are infinite and therefore vacuously
  covered by the specification above.
\end{example}

We propose an effective proof system that, given a LCTRS, proves
valid reachability formulas such as the one above, assuming an oracle
that solves logical constraints. In practice, we use an SMT solver
instead of the oracle.

\paragraph{Contributions}
\begin{enumerate*}
\item As computations can be finite or infinite, an inductive approach
  for reachability is not practically possible. In
  Section~\ref{sec:reach-prop}, we propose a coinductive approach for
  specifying transition systems, which is an elegant way to look at
  reachability, but also essential in
  handling both finite and infinite executions.
\item We formalize the semantics of LCTRSs as a reduction relation
  over a particular model that combines order-sorted terms with
  builtin elements such as integers, booleans, arrays, etc. The new
  approach, introduced in Section~\ref{sec:lctrss}, is simpler than
  the usual semantics for constrained term rewriting
  systems~\cite{DBLP:journals/corr/Kop16,DBLP:conf/frocos/KopN13,DBLP:conf/lpar/Kop015,tocl-lctrs},
  but it also lifts several technical restrictions that are important
  for our case studies.
\item We introduce a sound and complete coinductive proof system for
  deriving valid reachability formulas for transition systems
  specified by a LCTRS. We present our proof system in two steps: in
  the first step, we provide a three-rule proof
  system~(Figure~\ref{fig:symbolic-execution}) for \emph{symbolic
    execution of constrained terms}. When interpreting the proof
  system coinductively, its proof trees can be finite or infinite. The
  finite proof trees correspond to reachability formulas
  $\rform{\ct{t}{\phi}}{\ct{t'}{\phi'}}$ where there is a bounded
  number of symbolic steps between $\ct{t}{\phi}$ and
  $\ct{t'}{\phi'}$. The infinite proof trees correspond to proofs of
  reachability formulas $\rform{\ct{t}{\phi}}{\ct{t'}{\phi'}}$ that
  hold for an unbounded number of symbolic steps between
  $\ct{t}{\phi}$ and $\ct{t'}{\phi'}$ (obtained, e.g., by unrolling
  loops). Symbolic execution has similarities to narrowing, but unlike
  narrowing, where each step computes a possible successor, symbolic
  execution must consider \emph{all} successors of a state
  at~the~same~time.
\item The infinite proof trees above cannot be obtained in finite time
  in practice. In order to derive reachability formulas that require
  an unbounded number of symbolic steps in finite time, we introduce a
  fourth proof rule to the system that we call \emph{circularity}.
  The circularity proof rule can be used to compress infinite proof
  trees into finite proof trees. The intuition is to use as axioms the
  goals that are to be proven, when they satisfy a \emph{guardedness}
  condition. This compression of infinite coinductive trees into
  finite proof trees via the guardedness condition nicely complements
  our coinductive approach. This separation between symbolic execution
  and circularity answers an open question in~\cite{JSC2016}.
\item We introduce the {\tt RMT} tool, an implementation of the proof
  system that validates our approach on a number of examples. {\tt
    RMT} uses an SMT solver to discharge logical constraints. The tool
  is expressive enough for specifying various transition systems,
  including operational semantics of programming languages, and
  proving reachability properties of practical interest and is
  intended to be the starting point of a library for rewriting modulo
  builtins, which could have more applications.
\end{enumerate*}

\paragraph{Related Work}%
A number of
approaches~\cite{Aguirre,guarded,tocl-lctrs,Rocha2016,SkeirikSM2017}
to combining rewriting and SMT solving have appeared lately. The
rewrite tool Maude~\cite{DBLP:conf/cade/DuranEEMMT16} has been
extended with SMT solving in~\cite{Rocha2016} in order to enable the
analysis of open systems. A method for proving invariants based on an
encoding into reachability properties is presented
in~\cite{SkeirikSM2017}. Both approaches above are restricted to
\emph{topmost rewrite theories}. While almost any theory can be
written as a topmost theory~\cite{Meseguer:2007:SRA:1236620.1236625},
the encoding can significantly increase the number of transitions,
which raises performance concerns. Our definition for constrained term
is a generalization of that of constructor constrained pattern used
in~\cite{SkeirikSM2017}. In particular~\cite{SkeirikSM2017} does not
allow for quantifiers in constraints, but quantifiers are critical to
obtaining a complete proof system, as witnessed by their use in the
subsumption rule in our proof system ({\sf [subs]},
Figure~\ref{fig:symbolic-execution}). The approach without quantifiers
is therefore not sufficient to prove reachabilities in a general
setting.

A calculus for reachability properties in a formalism similar to
LCTRSs is given in~\cite{Aguirre}. However, the notion of reachability
in~\cite{Aguirre} is different from ours: while we show reachability
along \emph{all terminating paths of the computation},~\cite{Aguirre}
solves reachability properties of the form
$\exists \widetilde{x}.t(\widetilde{x}) \rightarrow^* t'(\widetilde{x})$
(i.e. does there exists an instance of $t$ that reaches, along some
path, an instance of $t'$).

Work on constrained term rewriting systems appeared
in~\cite{DBLP:conf/frocos/KopN13,DBLP:conf/lpar/Kop015,DBLP:journals/corr/Kop16,tocl-lctrs}. In
contrast to this approach to constrained rewriting, our semantics is
simpler (it does not require two reduction relations), it does not
have restrictions on the terms $l, r$ in a rule $\rrule{l}{r}{\phi}$
and the constraint is an arbitrary first-order formula $\phi$,
possibly with quantifiers, which are crucial to obtain symbolic
execution in its full generality. Constrained terms are generalized to
guarded terms in~\cite{guarded}, in order to reduce the state space.

Reachability in rewriting is explored in depth
in~\cite{ESCOBAR20075}. The work by Kirchner and
others~\cite{kirchner} is the first to propose the use of rewriting
with symbolic constraints for deduction. Subsequent
work~\cite{Rocha2016,DBLP:conf/frocos/KopN13,tocl-lctrs} extends and
unifies previous approaches to rewriting with constraints. The related
work section in~\cite{Rocha2016} includes a comprehensive account of
literature related to rewriting modulo constraints.

Our previous
work~\cite{stefanescu-ciobaca-mereuta-moore-serbanuta-rosu-2014-rta,JSC2016}
on proving program correctness was in the context of the K
framework~\cite{rosu-serbanuta-2010-jlap}. K, developed by Ro\c{s}u
and others, implements semantics-based program
verifiers~\cite{stefanescu-park-yuwen-li-rosu-2016-oopsla} for any
language that can be specified by a rewriting-based operational
semantics, such as C~\cite{DBLP:conf/pldi/HathhornER15},
Java~\cite{DBLP:conf/popl/BogdanasR15} and
JavaScript~\cite{DBLP:conf/pldi/ParkSR15}. Our formalism is not more
expressive than that of reachability
logic~\cite{stefanescu-ciobaca-mereuta-moore-serbanuta-rosu-2014-rta}
for proving partial correctness of programs in a language-independent
manner, but it does have several advantages. Firstly, we make a clear
separation between \emph{rewrite rules} (used to define transition
systems), for which it makes no sense to have constraints on both the
lhs and the rhs, and \emph{reachability formulas} (used to specify
reachability properties), for which there can be constraints on both
the lhs and the rhs. We provide clear semantics of both syntactic
constructs above, which makes it unnecessary to check well-definedness
of the underlying rewrite system, as required
in~\cite{stefanescu-ciobaca-mereuta-moore-serbanuta-rosu-2014-rta}. Additionally,
this separation, which we see as a contribution, makes it easy to get
rid of the top-most restriction in previous approaches. Another
advantage is that the proposed proof system is very easy to automate,
while being sufficiently expressive to specify real-world
applications. Additionally, we work in the more general setting of
LCTRSs, not just language semantics, which enlarges the possible set
of applications of the technique. We also have several major technical
improvements compared to~\cite{JSC2016}, where the proof system is
restricted to the cases where unification can be reduced to matching
and topmost rewriting. The totality property required for languages
specifications, which was quite restrictive, was replaced by a local
property in proof rules and all restrictions needed to reduce
unification to matching were removed.

In contrast to the work on partial correctness
in~\cite{stefanescu-park-yuwen-li-rosu-2016-oopsla}, the approach on
reachability discussed here is meant for any LCTRS, not just
operational semantics. The algorithm
in~\cite{stefanescu-park-yuwen-li-rosu-2016-oopsla} contains a small
source of incompleteness, as when proving a reachability property it
is either discharged completely through implication or through
circularities/rewrite rules. We allow a reachability rule to be
discharged partially by subsumption and partially by other
means. Constrained terms are a fragment of Matching Logic
(see~\cite{rosu-2017-lmcs}), where no distinction is made between
terms and constraints. Coinduction and circular or cyclic proofs have
been proposed in other contexts. For example, circular proof systems
have been proposed for first-order logic with inductive predicates
in~\cite{brotherston} and for separation logic
in~\cite{brotherston-aplas2012}. In the context of interactive theorem
provers, circular coinduction has been proposed as an incremental
proof method for bisimulation in process calculi
(see~\cite{popescu-fossacs2010}). A compositional and incremental
approach to coinduction that uses a semantic guardedness check instead
of a syntactic check is given in~\cite{hur-popl2013}.

\paragraph{Paper Structure} We present coinductive definitions for
execution paths and reachability predicates in
Section~\ref{sec:reach-prop}. In Section~\ref{sec:lctrss}, we
introduce logically constrained term rewriting with builtins in an
order-sorted setting. In Section~\ref{sec:proving}, we propose a sound
and complete coinductive calculus for reachability and a circularity
rule for compressing infinite proof trees into finite proof
trees. Section~\ref{sec:conclusion} discusses the implementation before
concluding. The proofs and a discussion of coinduction and
order-sorted algebras can be found in the Appendix.

\section{Reachability Properties: Coinductive Definition}
\label{sec:reach-prop}

In this section we introduce a class of reachability properties,
defined coinductively. A \emph{state predicate} is a subset of
states. A \emph{reachability property} is a pair $\rform{P}{Q}$ of
state predicates. Such a reachability property is \emph{demonically
  valid} iff each execution path starting from a state in $P$
eventually reaches a state in $Q$, or if it is infinite. Since the set
of finite and infinite executions is coinductively defined, the set of
valid predicates can be defined coinductively as well. Formally,
consider a transition system $(M,\gtran{})$, with
${\gtran{}}\subseteq M\times M$. We write $\gamma\gtran{} \gamma'$ for
$(\gamma,\gamma')\in{\gtran{}}$. An element $\gamma\in M$ is
\emph{irreducible} if $\gamma\not\gtran{}\gamma'$ for any
$\gamma' \in M$.

\begin{definition}[Execution Path] The set of \emph{(complete)
    execution paths} is coinductively defined by the following rules:
\\
  \centerline{$\dfrac{}{~\gamma~}~\gamma\in M, \gamma~{\rm irreducible} \qquad
  \dfrac{\tau}{\gamma_0\gpath{}\tau}~\gamma_0\gtran{}\hd(\tau)$}
\\
where the function $\hd$ is defined by
$\hd(\gamma)=\gamma$ and $\hd(\gamma_0\gpath{}\tau)=\gamma_0$.
\end{definition}

The above definition includes both the finite execution paths ending
in a irreducible state and the infinite execution paths, defined as
the greatest fixed point of the associated functional (see
Appendix~\ref{sec:coind}).

\begin{definition}[State and Reachability Predicates] A \emph{state
    predicate} is a subset $P\subseteq M$. A \emph{reachability
    predicate} is a pair of state predicates $\rform{P}{Q}$. The
  predicate $P$ is \emph{runnable} if $P \not= \emptyset$ and for all
  $\gamma\in P$ there is $\gamma'\in M$ s.t.  $\gamma\gtran{}\gamma'$.
\end{definition}
A \emph{derivative} measures the sensitivity to change of a
quantity. For the case of transition systems, the change of states is
determined by the transition relation.
\begin{definition}[Derivative of a State Predicate] The derivative of
  a state predicate $P$ is the state predicate
  $\partial(P) = \{\gamma' \mid
  \gamma\gtran{}\gamma'\textrm{~for~some~}\gamma\in P\}$.
\end{definition}
As a reachability predicate specifies reachability property of
execution paths, we define when a particular execution path satisfies
a reachability predicate.
\begin{definition}[Satisfaction of a Reachability Predicate] An
  execution path $\tau$ \emph{satisfies} a reachability predicate
  $\rform{P}{Q}$, written $\tau\demModels\rform{P}{Q}$, iff
  $\langle \tau,\rform{P}{Q}\rangle\in \nu\,\widehat{\sf EPSRP}$,
  where {\sf EPSRP} consists of the following rules:
  \\
  \centerline{$ \dfrac{}{~\langle \tau, \rform{P}{Q}\rangle}~\hd(\tau)\in P\cap Q
    \qquad
    \dfrac{\langle \tau,\rform{\partial(P)}{Q}\rangle}{\langle
      \gamma_0\gpath{}\tau,\rform{P}{Q}\rangle}~\gamma_0\in P,
    \gamma_0\gtran{}\hd(\tau). $}
\end{definition}
The notation $\widehat{\sf EPSRP}$ stands for the functional of {\sf
  EPSRP} and $\nu\widehat{\sf EPSRP}$ stands for its greatest fixed
point (see Appendix~\ref{sec:coind}). We coinductively define the set
of \emph{demonically valid reachability predicates} over
$(M,\gtran{})$. This allows to use coinductive proof techniques to
prove validity of reachability predicates.
\begin{definition}[Valid Reachability Predicates, 
  Coinductively] \label{def:dvp} We say that $\rform{P}{Q}$ is
  \emph{demonically valid}, and we write \\
  \centerline{$(M,\gtran{})\demModels\rform{P}{Q},$}
  iff
  $\rform{P}{Q}\in \nu\,\widehat{\sf DVP}$, where {\sf DVP} consists
  of the following rules:\\
  \centerline{$
  {\semrule{Subsumption}}~\dfrac{}{\rform{P}{Q}}~P\subseteq Q
  \qquad
  {\semrule{Step}}~\dfrac{\rform{\partial(P\setminus
      Q)}{Q}}{\rform{P}{Q}}~P\setminus Q~{\rm runnable}.$}
\end{definition}
The condition \emph{$P \setminus Q$ runnable} in the second rule is
essential to avoid the cases where execution is stuck. These blocking
states have no successor in $\partial(P \setminus Q)$ and, in the
absense of the condition, we would wrongly conclude that they satisfy
$\rform{P}{Q}$. The terminating executions are captured by
$\semrule{Subsumption}$.

The following proposition justifies our definition of demonically
valid reachability predicates.
\begin{restatable}{proposition}{demonicvalidsoundness} Let
  $\rform{P}{Q}$ be a reachability predicate. We have
  $(M,\gtran{})\demModels\rform{P}{Q}$ iff any execution path $\tau$
  \emph{starting from $P$} ($\hd(\tau) \in P$) satisfies
  $\rform{P}{Q}$.
\end{restatable}

\section{Logically Constrained Term Rewriting Systems}
\label{sec:lctrss}

In this section we introduce our formalism for LCTRSs. We interpret
LCTRSs in a model combining order-sorted terms with builtins such as
integers, booleans, etc. Logical constraints are first-order formulas
interpreted over the fixed model.

We assume a \emph{builtin model} $\bltn{M}$ for a many-sorted
\emph{builtin signature} $\bltn{\Sigma} = (\bltn{S}, \bltn{F})$, where
$\bltn{S}$ is a set of \emph{builtin sorts} that includes at least the
sort $\Bool$ and $\bltn{F}$ is the $\bltn{S}$-sorted set of
\emph{builtin function symbols}. We assume that the set interpreting
the sort $\Bool$ in the model $\bltn{M}$ is $\bltn{M}_{\Bool} = \{
\True, \False \}$. We use the standard notation $M_o$ for the
interpretation of the sort/symbol $o$ in the model $M$. The set
$\bltn{\CF}$, defined as the set of (many-sorted) first-order formulas
with equality over the signature $\bltn{\Sigma}$, is the set of
\emph{builtin constraint formulas}. Functions returning $\Bool$ play
the role of predicates and terms of sort $\Bool$ are atomic
formulas. We will assume that the builtin constraint formulas can be
decided by an oracle (implemented as an SMT solver).

A \emph{signature modulo builtins} is an order-sorted signature
$\Sigma = (S, \le, F)$ that includes $\bltn{\Sigma}$ as a subsignature
and such that the only builtin constants in $\Sigma$ are elements of
the builtin model ($\{ c \mid c \in F_{\varepsilon,s}, s \in \bltn{S}
\} = \bltn{M}_s$) -- therefore the signature might be infinite. By
$F_{w,s}$ we denoted the set of function symbols of arity $w$ and
result sort $s$. $\bltn{\Sigma}$ is called the \emph{builtin
  subsignature} of $\Sigma$ and $\cons{\Sigma}=(S,\le, (F\setminus
\bltn{F})\cup\bigcup_{s\in \bltn{S}}F_{\varepsilon,s})$ the
\emph{constructor subsignature} of $\Sigma$. We let $\X$ be an
$S$-sorted set of variables.

We extend the builtin model $\bltn{M}$ to an $(S,\le,F)$-model
$M^\Sigma$ defined as follows:
$\bullet$ $M^\Sigma_s = T_{\cons{\Sigma},s}$, for each
$s \in S \setminus \bltn{S}$ ($M^\Sigma_s$ is the set of ground
constructor terms of sort $s$, i.e. terms built from constructors
applied to builtin elements);
$\bullet$ $M^\Sigma_f=\bltn{M}_f$ for each builtin function symbol
$f\in\bltn{F}$;
$\bullet$ $M^\Sigma_f$ is the term constructor
$M^\Sigma_f(t_1,\ldots,t_n) = f(t_1,\ldots,t_n)$, for each
non-builtin function symbol $f \in F \setminus \bltn{F}$.
By fixing the interpretation of the non-builtin function symbols, we
can reduce constraint formulas to built-in constraint formulas by
relying on an unification algorithm described in detail
in~\cite{wollic-2018}. We also make the standard assumption that
$M_s \not= \emptyset$ for any $s \in S$.

\begin{example}\label{ex:model-generated}%
  Let $\bltn{\Sigma} = (\bltn{S}, \bltn{F})$, where
  $\bltn{S} = \{ \Int, \Bool \}$ and $\bltn{F}$ include the usual
  operators over $\Bool$eans ($\lor, \land, \ldots$) and over the
  $\Int$egers ($+, -, \times, \ldots$). The builtin model $\bltn{M}$
  interprets the above sorts and operations as expected.

  We consider the signature modulo builtins $\Sigma = (S, \le, F)$,
  where the set of sorts $S = \{ \Cfg, \Int, \Bool \}$ consists of the
  builtin sorts and an additional sort $\Cfg$, where the subsorting
  relation ${\le} \subseteq S \times S = \emptyset$ is empty, 
  and where the set of function symbols $F$ includes, in addition to the
  builtin symbols in $\bltn{F}$, the following function symbols:
  $\textit{init} : \Int \to \Cfg, \textit{loop} : \Int \times \Int \to
  \Cfg, \textit{comp} : \Cfg$. We have that
  $M^\Sigma_\Cfg = \{ \textit{init}(i) \mid i \in \mathbb{Z} \} \cup
  \{ \textit{loop}(i,j) \mid i, j \in \mathbb{Z} \} \cup \{
  \textit{comp} \}$.
\end{example}

The set $\CF$ of \emph{constraint formulas} is the set of first-order
formulas with equality over the signature $\Sigma$. 
The subset of the builtin constraint formulas is denoted by $\CF^b$. 
Let $\var(\phi)$ denote the set of variables freely
occurring in $\phi$. We write $M^\Sigma, \alpha \models \phi$ when the
formula $\phi$ is satisfied by the model $M^\Sigma$ with a valuation
$\alpha : X \rightarrow M^\Sigma$.

\begin{example}\label{ex:constraint-formula}
  The constraint formula
  $\phi \eqbydef {\exists u.1 < u < n \land n
    \allowbreak\mathrel{\textit{mod}} u = 0}$ is satisfied by the
  model $M^\Sigma$ defined in Example~\ref{ex:model-generated} and any
  valuation $\alpha$ such that $\alpha(n)$ is a composite number.

\end{example}

\begin{definition}[Constrained Terms]
  A \emph{constrained term} $\varphi$ of sort $s\in S$ is a pair
  $\ct{t}{\phi}$, where $t\in T_{\Sigma,s}(\X)$ and $\phi \in \CF$.

\end{definition}

\begin{example}%
\label{ex:constrained-term}%
Continuing the previous example, the following is a constrained term:
$\ct{\textit{init}(n)}{\exists u.1 < u < n \land n
  \mathrel{\textit{mod}} u = 0}.$
\end{example}

We consistently use $\varphi$ for constrained terms and $\phi$ for
constraint formulas.

\begin{definition}[Valuation Semantics of Constraints]
\label{def:valuationsemantics}
The \emph{valuation semantics} of a constraint $\phi$ is the set
$\sem{\phi} \eqbydef \{\alpha : X \to M^\Sigma \mid M^\Sigma, \alpha
\models \phi\}$.
\end{definition}
\begin{example}\label{ex:constrained-term-valuation-semantics}
  Continuing the previous example, we have that
  \\
  \centerline{$\begin{array}{l}\sem{\exists u.1
                 < u < n \land n \mathrel{\textit{mod}} u = 0} = 
      \{ \alpha : X \to M^\Sigma \mid \mbox{$\alpha(n)$ is
        composite}\}.\end{array}$}
\end{example}
\begin{definition}[State Predicate Semantics of Constrained
  Terms]
\label{def:statepredicatesemantics}
The \emph{state predicate semantics} of a constrained term
$\ct{t}{\phi}$ is the set
\\
\centerline{$\tsem{\ct{t}{\phi}}\eqbydef\{\alpha(t)\mid\alpha\in\sem{\phi}\}.$}
\end{definition}
\begin{example}
\label{ex:constrained-term-state-predicate-semantics}
Continuing the previous example, we have that
\\
\centerline{$\tsem{\ct{\textit{init}(n)}{\exists u.1 < u < n \land n
      \mathrel{\textit{mod}} u = 0}} = \{ \textit{init}(n) \mid
  \mbox{$n$ is composite} \}.$}
\end{example}
We now introduce our formalism for logically constrained term
rewriting systems. Syntactically, a rewrite rule consists of two terms
(the left hand side and respectively the right hand side), together
with a constraint formula. As the two terms could \emph{share some
  variables}, these shared variables should be instantiated
consistently in the semantics:

\begin{definition}[LCTRS]
\label{d:rules}
A \emph{logically constrained rewrite rule} is a tuple $(l,r,\phi)$,
often written as $\rrule{l}{r}{\phi}$, where $l,r$ are terms in
$T_{\Sigma}(\X)$ having the same sort, and $\phi \in \CF$. A
\emph{logically constrained term rewriting system} $\R$ is a set of
logically constrained rewrite rules. $\R$ defines an order-sorted
\emph{transition relation} $\gtran{\R}$ on $M^\Sigma$ as follows:
$t \gtran{\R} t'$ iff there exist a rule $\rrule{l}{r}{\phi}$ in $\R$,
a context $c[\cdot]$, and a valuation $\alpha : X \to M^\Sigma$ such
that $t = \alpha(c[l])$, $t' = \alpha(c[r])$ and
$M^\Sigma,\alpha \models \phi$.
\end{definition}
\begin{example}\label{ex:constrained-rule-system}
  We recall the LCTRS given in the introduction:\\
  \centerline{$\R = \left\{\begin{array}{l}
\rrule{\textit{init}(n)}{\textit{loop}(n, 2)}{\True}, \\
\rrule{\textit{loop}(i \times k, i)}{\textit{comp}}{k > 1}, \\
\rrule{\textit{loop}(n, i)}{\textit{loop}(n, i + 1)}{\lnot(\exists k.k
      > 1 \land n = i \times k)}
\end{array}\right\}.$}
\end{example}

A LCTRS $\R$ defines a sort-indexed transition system
$(M^\Sigma,\gtran{\R})$. As each constrained term $\varphi$ defines a
state predicate $\tsem{\varphi}$, it is natural to specify
reachability predicates as pairs of constrained terms sharing a subset
of variables. The shared variables must be instantiated in the same
way by the execution paths connecting states specified by the two
constrained terms.

\begin{definition}[Reachability Properties of LCTRSs]
  A \emph{reachability formula} $\rform{\varphi}{\varphi'}$ is a pair
  of constrained terms, which may share variables. We say that a LCTRS
  \emph{$\R$ demonically satisfies $\rform{\varphi}{\varphi'}$},
  written
  \\
  \centerline{$\R\demModels \rform{\varphi}{\varphi'},$}
  \\
  iff
  $(M^\Sigma,\gtran{\R})\demModels
  \rform{\tsem{\sigma(\varphi)}}{\tsem{\sigma(\varphi')}}$ for each
  $\sigma:\var(\varphi)\cap\var(\varphi')\to M^\Sigma$.

\end{definition}

Since the carriers sets of $M^\Sigma$ consist of ground terms,
$\sigma$ is both a substitution and a valuation in the definition
above. Its role is critical: to ensure that the shared variables of
$\varphi$ and $\varphi'$ are instantiated by the same values.

\begin{example}
  Continuing the previous example, we have that the reachability
  formula
  $\rform{\ct{\textit{init}(n)}{\exists u.1 < u < n \land n
      \mathrel{\textit{mod}} u = 0}}{\ct{\textit{comp}}{\True}}$
  is demonically satisfied by the constrained rule system $\R$ defined
  in Example~\ref{ex:constrained-rule-system}:
  \\
  \centerline{$\R \demModels \rform{\ct{\textit{init}(n)}{\exists u.1
        < u < n \land n \mathrel{\textit{mod}} u =
        0}}{\ct{\textit{comp}}{\True}}.$}
  We have checked the above reachability formula against $\R$
  mechanically, using an implementation of the approach described in
  this paper.

\end{example}

\section{Proving Reachability Properties of LCTRSs}
\label{sec:proving}

We introduce two proof systems for proving reachability properties in
transition systems specified by LCTRSs. The first proof system
formalizes \emph{symbolic execution} in a LCTRS, in the following
sense: a reachability formula $\rform{\varphi}{\varphi'}$ can be
proven if either the left-hand side $\varphi$ can be derived
infinitely many times (and therefore all execution paths starting with
$\varphi$ are infinite), or if some derivative is an instance of the
right-hand side $\varphi'$, i.e. all the execution paths starting with
$\varphi$ reach a state that is an instance of $\varphi'$. Note that
this intuition holds when the proof system is interpreted
coinductively, where infinite proof trees are allowed. Unfortunately,
these infinite proof trees have a limited practical use because they
cannot be obtained in finite time.

In order to solve this limitation, we introduce a second proof system,
which contains an additional inference rule, called
\emph{circularity}. The \emph{circularity} rule allows to use the
reachability formula to be proved as an axiom. This allows to fold
infinite proof trees into finite proof trees, which can be obtained in
finite time. Adding the reachability formulas that are to be proved as
axioms seems at first to be unsound, but it corresponds to a natural
intuition: when reaching a proof obligation that we have handled
before, there is no need to prove it again, because the previous
reasoning can be reused (possibly leading to an infinite proof
branch). However, the circularity rule must be used in a guarded
fashion in order to preserve soundness. We introduce a simple
criterion to select the sound proof trees.

\subsection{Derivatives of Constrained Terms}

Our proof system relies on the notion of derivative at the syntactic
level:
\begin{definition}[Derivatives of Constrained Terms]
  The \emph{set of derivatives} of a constrained term
  $\varphi\eqbydef \ct{t}{\phi}$ w.r.t. a rule
  $\rrule{l}{r}{\phi_{\it lr}}$ is
\begin{align}
\Delta_{l,r,\phi_{\it lr}}(\varphi) \eqbydef \{\ct{c[r]}{\phi'}\mid {}
&\phi' \eqbydef \phi \land t \qeq c[l] \land \phi_{\it lr},\notag\\
& c[\cdot]\textrm{~an~appropriate~context~and~}\phi'\textrm{~is~satisfiable}\},\label{eq:der}
\end{align}
\noindent where the variables in $\rrule{l}{r}{\phi_{\it lr}}$ are renamed such
that $\var(l,r,\phi_{\it lr})$ and $\var(\varphi)$ are disjoint. If
$\R$ is a set of rules, then
$\Delta_\R(\varphi)=\bigcup_{(l,r,\phi_{\it lr})\in
  \R}\Delta_{l,r,\phi_{\it lr}}(\varphi).$ A constrained term
$\varphi$ is \emph{$\R$-derivable} if
$\Delta_{\R}(\varphi)\not=\emptyset$.

\end{definition}

\begin{example}
  Continuing the previous examples, we have that
  \\
  \centerline{$\begin{array}{l}\Delta_\R(\ct{\textit{init}(n)}{\exists
                 u.1 < u < n \land n \mathrel{\textit{mod}} u = 0}) = \\
                 \qquad\!\! \{ \ct{\textit{loop}(n, 2)}{\exists u.1 <
                 u < n \land n \mathrel{\textit{mod}} u = 0}
                 \}.\end{array}$}
In the above case, $\Delta_\R$ includes only the derivative computed
w.r.t. the first rule in $\R$, because the constraints of the ones
computed w.r.t. the other rules are unsatisfiable. Intuitively, the
derivatives of a constrained term denote all its possible successor
configurations in the transition system generated by $\R$.
\end{example}

The symbolic derivatives and the concrete ones are related as
expected:
\begin{restatable}{theorem}{deltacommutes}
\label{th:der}
Let $\varphi \eqbydef \ct{t}{\phi}$ be a constrained term, $\R$ a
constrained rule system, and $(M^\Sigma,\gtran{\R})$ the transition
system defined by $\R$.  Then
$\tsem{\Delta_\R(\varphi)} = \partial(\tsem{\varphi})$.
\end{restatable}

Our proof systems allows to replace any reachability formula by an
equivalent one. Two reachability formulas,
$\rform{\varphi_1}{\varphi'_1}$ and $\rform{\varphi_2}{\varphi'_2}$,
are \emph{equivalent}, written
$\rform{\varphi_1}{\varphi'_1}\equiv\rform{\varphi_2}{\varphi'_2}$,
if, for all LCTRSs $\R$,
\\
\centerline{$\R\demModels \rform{\varphi_1}{\varphi'_1}$ iff
  $\R\demModels\rform{\varphi_2}{\varphi'_2}.$}

We write $\tsem{\varphi}\subseteq_{\it shared}\tsem{\varphi'}$ iff for
each $\sigma : \var(\varphi) \cap \var(\varphi') \to M^\Sigma$, we
have $\tsem{\sigma(\varphi)}\subseteq\tsem{\sigma(\varphi')}$. The
next result, used in our proof system, shows that inclusion of the
state predicate semantics of two constrained terms can be expressed as
a constraint formula, when the shared variables are instantiated
consistently.
\begin{restatable}{proposition}{semanticinclusion}\label{prop:subs}
The inclusion $\tsem{\ct{t}{\phi}}\subseteq_{\it shared}\tsem{\ct{t'}{\phi'}}$
holds if and only if
$M^\Sigma\models \phi\limplies(\exists \widetilde{x})(t \qeq t'\land \phi')$,
where $\widetilde{x} \eqbydef \var(t',\phi')\setminus\var(t,\phi)$.
\end{restatable}
 
\subsection{Proof System for Symbolic Execution}

The first proof system, ${\sf DSTEP}$, derives sequents of the form
$\rform{\ct{t_l}{\phi_l}}{\ct{t_r}{\phi_r}}$. The proof system
consists of three proof rules presented in
Figure~\ref{fig:symbolic-execution} and an implicit structural rule
that allows to replace reachability formulas by equivalent
reachability formulas. The instances of this implicit structural rule
are not included in the proof trees.  We explain the three rules in
the proof system.
\begin{figure}[t]
\begin{align*}
&[{\sf
  axiom}]~\dfrac{}{\rform{\ct{t_l}{\False}}{\ct{t_r}{\phi_r}}}\\
&[{\sf subs}]~\dfrac{\rform{\ct{t_l}{\phi_l \land \neg (\exists{\tilde{x}}.t_l \qeq t_r \land \phi_r)}}{\ct{t_r}{\phi_r}}}{\rform{\ct{t_l}{\phi_l}}{\ct{t_r}{\phi_r}}}
  ~~\begin{array}{l}
      {\widetilde{x}} \eqbydef \var(t_r,\phi_r) \setminus
      \var(t_l,\phi_l)\\
      \exists{\tilde{x}}.t_l \qeq t_r \land \phi_r\mbox{ satisfiable}
     \end{array}
\\
&[{\sf der}^\forall]~
  \dfrac{
  \rform{\ct{t^j}{\phi^j}}{\ct{t_r}{\phi_r}}, j \in \{ 1, \ldots, n\}
  }{
  \rform{\ct{t_l}{\phi_l}}{\ct{t_r}{\phi_r}}
  }~~\begin{array}{l}
\ct{t_l}{\phi_l}{\rm~is~}\R {\rm-derivable~and}\\
\phi_l \limplies \bigvee_{j \in \{1, \ldots, n\}}\exists{\widetilde{y}^j}.\phi^j{\rm~is~valid}
     \end{array}\\
&\phantom{[{\sf der}^\forall]~}
\begin{array}{l}
{\rm where~}\Delta_\R(\ct{t_l}{\phi_l}) = \{ \ct{t^1}{\phi^1}, \ldots,
  \ct{t^n}{\phi^n} \}{\rm~and}\\
{\rm \phantom{where~}}\widetilde{y}^j = \var(t^j, \phi^j) \setminus \var(t_l, \phi_l)
\end{array}
\end{align*}
\caption{\label{fig:symbolic-execution}The ${\sf DSTEP}(\R)$ Proof System}
\end{figure}
\\
\noindent
$\bullet$ The [{\sf axiom}] rule discharges goals where the left hand
side of the goal does not match any state. As our structural rule
identifies equivalent reachability formulas, this rule can be applied
to any left-hand side where the constraint is unsatisfiable
(equivalent to $\False$). This rule discharges reachability formulas
where there are no execution paths starting from the left-hand side,
and therefore no need to continue the proof process.
\\
$\bullet$ The [{\sf subs}] rule discharges the cases where the
left-hand side is an instance of the right-hand side. The constraint
$\exists{\widetilde{x}}.t_l \qeq t_r \land \phi_r$ is true exactly when
the left-hand side is an instance of the right-hand side, which is
ensured by Proposition~\ref{prop:subs}. The proof of the current goal continues
only for the cases where the negation of this constraint holds (i.e., the cases where the left-hand
side is not included in the right-hand side).
\\
$\bullet$ The [{\sf der}$^\forall$] rule allows to take a symbolic
step in the left-hand side of the current goal. It computes all
derivatives of the left-hand side; the proof process must continue
with each such derivative. Let
$\psi \eqbydef \phi_l \limplies{}
\bigvee\{\exists{\widetilde{y}^j}.\phi^j\}$ be the logical constraint
that occurs in the condition of [{\sf der}$^\forall$].  The formula
$\psi$ is valid iff there is at least one rule of $\R$ that can be
applied to any instance of $\ct{t_l}{\phi_l}$, meaning that $\R$ is
\emph{total for $\ct{t_l}{\phi_l}$}. Summarising, the condition of
[{\sf der}$^\forall$] says that $\tsem{\ct{t_l}{\phi_l}}$ must have at
least one successor and furthermore that any instance
$\gamma \in \tsem{\ct{t_l}{\phi_l}}$ has a $\gtran{\R}$-successor.

The following result shows that ${\sf DSTEP}(\R)$ is sound and
complete.

\begin{restatable}{theorem}{dstepsound}
\label{th:dstep}
Let $\R$ be a LCTRS. For any reachability formula
$\rform{\varphi}{\varphi'}$, we have
\\
\centerline{$\R \demModels \rform{\varphi}{\varphi'}$ iff
$\rform{\varphi}{\varphi'} \in \nu\,\widehat{\sf DSTEP}(\R)$.}

\end{restatable}

\begin{example}
\label{ex:symbexec}
Consider the LCTRS $\R$ defined in
Example~\ref{ex:constrained-rule-system}.  The proof tree for the
reachability formula $\rform{\ct{\textit{init}(n)}{\psi}}{\varphi_r},$
where
$\psi \eqbydef \exists u.1 < u < n \land n \mathrel{\textit{mod}} u =
0$ denotes the fact that $n$ is composite and
$\varphi_r \eqbydef \ct{\textit{comp}}{\True}$, is infinite:
\begin{prooftree}\fontsize{9}{11}\selectfont
\AxiomC{}
\RightLabel{[{\sf axiom}]}
\UnaryInfC{$\rform{\ct{\textit{comp}}{\False}}{\varphi_r}$}
\RightLabel{[{\sf subs}]}
\UnaryInfC{$\rform{\ct{\textit{comp}}{\psi \land \phi_a}}{\varphi_r}$}
\AxiomC{}
\RightLabel{[{\sf axiom}]}
\UnaryInfC{$\rform{\ct{\textit{comp}}{\False}}{\varphi_r}$}
\RightLabel{[{\sf subs}]}
\UnaryInfC{$\rform{\ct{\textit{comp}}{\psi \land \phi_2 \land \phi_b}}{\varphi_r}$}
\AxiomC{\vdots}
\RightLabel{[${\sf der}^\forall$]}
\BinaryInfC{$\rform{\ct{\textit{loop}(n,3)}{\psi \land \phi_2}}{\varphi_r}$}
\RightLabel{[${\sf der}^\forall$]}
\BinaryInfC{$\rform{\ct{\textit{loop}(n,2)}{\psi}}{\varphi_r}$}
\RightLabel{[${\sf der}^\forall$]}
\UnaryInfC{$\rform{\ct{\textit{init}(n)}{\psi}}{\varphi_r}$}
\end{prooftree}
\noindent The right branch of the
above proof tree is infinite, and:
\\
\centerline{
$\begin{aligned}
&\phi_2 \eqbydef \lnot \exists k.k > 1 \land n = 2 \times k \qquad
\phi_a \eqbydef \textit{loop}(n, 2) = \textit{loop}(i' \times k', i')
\land k' > 1 \\
& \phi_3 \eqbydef \lnot \exists k.k > 1 \land n = 3 \times k \qquad
\phi_b \eqbydef \textit{loop}(n, 3) = \textit{loop}(i' \times k', i')
\land k' > 1 \\
& \qquad \qquad \qquad \qquad \qquad \qquad \qquad \ldots
\end{aligned}$}\\

Note that in the presentation of the tree above, we used the
structural rule to replace reachability formulas by equivalent
reachability formulas as follows:

\centerline{
  $\begin{aligned}
&
\rform{\ct{\textit{comp}}{\psi \land \phi_a \land \neg (\textit{comp} = \textit{comp} \land \True)}}{\varphi_r}
& \equiv & \;\;\;
\rform{\ct{\textit{comp}}{\False}}{\varphi_r}, \\
&
\rform{\ct{\textit{comp}}{\psi \land \phi_2 \land \phi_b \land \neg (\textit{comp} = \textit{comp} \land \True)}}{\varphi_r}
& \equiv &\;\;\;
\rform{\ct{\textit{comp}}{\False}}{\varphi_r}, \\
&
\rform{\ct{\textit{loop}(n',2)}{\True \land \textit{init}(n') =
  \textit{init}(n) \land \psi}}{\varphi_r}
& \equiv &\;\;\;
\rform{\ct{\textit{loop}(n,2)}{\psi}}{\varphi_r}, \\
&
\rform{\ct{\textit{loop}(n',i'+1)}{\psi \land \phi_2'}}{\varphi_r}
& \equiv &\;\;\;
\rform{\ct{\textit{loop}(n,3)}{\psi \land \phi_2}}{\varphi_r},
\end{aligned}$}
\noindent where
$\phi_2' \eqbydef \textit{loop}(n,2) = \textit{loop}(n', i') \land
\lnot \exists k.k > 1 \land n' = i' \times k$. The ticks appear in the
formulas above because, to compute derivatives, we used
the following fresh instance of $\R$:\\
\centerline{$\R = \left\{\begin{array}{l}
                           \rrule{\textit{init}(n')}{\textit{loop}(n', 2)}{\True}, \\
                           \rrule{\textit{loop}(i' \times k', i')}{\textit{comp}}{k' > 1}, \\
                           \rrule{\textit{loop}(n',
                           i')}{\textit{loop}(n', i' + 1)}
                           {\lnot(\exists k.k > 1 \land n' = i' \times
                           k)}
\end{array}\right\}.$}

\end{example}

\subsection{Extending the Proof System with a Circularity Rule}
 
As we said at the beginning of the section, the use of {\sf DSTEP} is
limited because of the infinite proof trees. The next inference rule
is intended to use the initial goals as axioms to fold infinite {\sf
  DSTEP}-proof trees into sound finite proof trees.

\begin{definition}[Demonic circular coinduction]
\label{def:ccsyst}
Let $G$ be a finite set reachability formulas. Then the set of rules
{\sf DCC}$(\R,G)$ consists of ${\sf DSTEP}(\R)$, together with
\[
[{\sf circ}]~
\dfrac{
\begin{array}{l}
\rform{\ct{t^c_r}{\phi_l \land \phi \land \phi^c_r}}{\varphi_r},\\
\rform{\ct{t_l}{\phi_l \land \neg \phi}}{\varphi_r}
\end{array}
}
{\rform{\ct{t_l}{\phi_l}}{\varphi_r}}~
~
\begin{array}{l}
\phi\mbox{ is }\exists \var(t^c_l,\phi^c_l).t_l \qeq t^c_l \land \phi^c_l,\\
\rform{\ct{t^c_l}{\phi^c_l}}{\ct{t^c_r}{\phi^c_r}} \in G
\end{array}
\] 
where $\rform{\ct{t^c_l}{\phi^c_l}}{\ct{t^c_r}{\phi^c_r}}$ is a rule
in $G$ whose variables have been renamed with fresh names.
\end{definition}

The idea is that $G$ should be chosen conveniently so that {\sf
  DCC}$(\R,G)$ proves $G$ itself. We call such goals $G$ (that are
used to prove themselves) \emph{circularities}. The intuition behind
the rule is that the formula $\phi$ defined in the rule holds when a
circularity can be applied. In that case, it is sufficient to continue
the current proof obligation from the rhs of the circularity
$\ct{t^c_r}{\phi^c_r \land \phi_l \land \phi}$. The cases when $\phi$
does not hold (the circularity cannot be applied) are captured by the
proof obligation $\rform{\ct{t_l}{\phi_l \land \neg
    \phi}}{\varphi_r}$.

Of course, not all proof trees under ${\sf DCC}(\R,G)$ are sound. The
next two definitions identify a class of sound proof trees
(cf. Theorem~\ref{th:cp}).

\begin{definition}
\label{def:ppt}
Let $\it PT$ be a proof tree of $\rform{\varphi}{\varphi'}$ under
${\sf DCC}(\R,G)$. A {\sf [circ]} node in $\it PT$ is \emph{guarded}
iff it has as ancestor a $\sf [der^\forall]$ node. $\it PT$ is
\emph{guarded} iff all its {\sf [circ]} nodes are guarded.

\end{definition}

\begin{definition}
\label{def:dement}
We write $(\R,G)\demEntails \rform{\varphi}{\varphi'}$ iff there is a
proof tree of $\rform{\varphi}{\varphi'}$ under ${\sf DCC}(\R,G)$ that
is guarded. If $F$ is a set of reachability formulas, we write
$(\R,G)\demEntails F$ iff $(\R,G)\demEntails\rform{\varphi}{\varphi'}$
for all $\rform{\varphi}{\varphi'} \in F$.

\end{definition}

The criterion stated by Definition~\ref{def:ppt} can be easily checked
in practice. The following theorem states that the guarded proof trees
under {\sf DCC} are sound.

\begin{restatable}[Circularity
  Principle]{theorem}{circularityprinciple} \label{th:cp} Let $\R$ be a
  constrained rule system and $G$ a set of goals. If
  $(\R,G) \demEntails G$ then $\R \demModels G$.

\end{restatable}

Theorem~\ref{th:cp} can be used by finding a set of circularities and
using them in a guarded fashion to prove themselves. Then the
circularity principle states that such circularities hold.

\begin{example}
  In order to prove
  $\rform{\ct{\textit{init}(n)}{\psi_1}}{\ct{\textit{comp}}{\True}}$,
  we choose the following set of circularities
  \\
  \centerline{$G =
    \left\{\begin{array}{l}\rform{\ct{\textit{init}(n)}{\psi_1}}{\ct{\textit{comp}}{\True}},
             \\
\rform{\ct{\textit{loop}(n, i)}{2 \leq i \land \exists u.i \leq u < n \land n \mathrel{\textit{mod}} u
= 0}}{\ct{\textit{comp}}{\True}}
\end{array}\right\}.$}%

The second circularity is inspired by the infinite branch of the proof tree under {\sf DSTEP}.
We will show that $(\R, G) \vdash^\forall G$, and by
Theorem~\ref{th:cp}, it follows that all reachability formulas in $G$
hold in $\R$.

\paragraph{First circularity.} To obtain a proof of the first
circularity,
$\rform{\ct{\textit{init}(n)}{\psi_1}}{\ct{\textit{comp}}{\True}}$, we
replace the infinite subtree rooted at
$\rform{\ct{\textit{loop}(n, 2)}{\psi}}{\varphi_r}$ 
in Example~\ref{ex:symbexec} by the following finite proof tree (that
uses [{\sf circ}]):
\begin{prooftree}\fontsize{9}{11}\selectfont
\AxiomC{}
\RightLabel{[{\sf axiom}]}
\UnaryInfC{$\rform{\ct{\textit{comp}}{\False}}{\varphi_r}$}
\RightLabel{[{\sf subs}]}
\UnaryInfC{$\rform{\ct{\textit{comp}}{\psi \land \phi\land \True}}{\varphi_r}$}
\AxiomC{}
\RightLabel{[{\sf axiom}]}
\UnaryInfC{$\rform{\ct{\textit{loop}(n,2)}{\psi \land \lnot \phi}}{\varphi_r}$}
\RightLabel{[{\sf circ}]}
\BinaryInfC{$\rform{\ct{\textit{loop}(n, 2)}{\psi}}{\varphi_r}$}
\end{prooftree}
where $\phi \eqbydef \exists n',i'.\textit{loop}(n, 2) = \textit{loop}(n',
i') \land 2 \leq i' \land \exists u.i' \leq u < n' \land n'
\mathrel{\textit{mod}} u = 0$.

\paragraph{Second circularity.} To complete the proof of $G$, we have
to find a finite proof tree for
\\
\centerline{$\rform{\ct{\textit{loop}(n, i)}{2 \leq i \land \exists
      u.i \leq u < n \land n \mathrel{\textit{mod}} u =
      0}}{\ct{\textit{comp}}{\True}}$}
\\
as well.
This is also obtained using [{\sf circ}] as follows:
\begin{prooftree}\fontsize{9}{11}\selectfont
\AxiomC{}
\RightLabel{[{\sf axiom}]}
\UnaryInfC{$
\begin{array}{ll}
\rform{\ct{\textit{comp}}{\False}}{\varphi_r}
\end{array}$}
\RightLabel{[{\sf subs}]}
\UnaryInfC{$\rform{\ct{\textit{comp}}{\psi_i \land \psi_a}}{\varphi_r}$}
\AxiomC{$T_1$ \qquad $T_2$}
\RightLabel{[${\sf circ}$]}
\UnaryInfC{$
\begin{array}{ll}
\rform{\ct{\textit{loop}(n,i + 1)}{\psi_i \land
  \psi_b}}{\varphi_r}
\end{array}$}
\RightLabel{[${\sf
         der}^\forall$]}
\BinaryInfC{$
\rform{\ct{\textit{loop}(n, i)}{\psi_i}}{\ct{\textit{comp}}{\True}}$}
\end{prooftree}
\noindent where
\\
\centerline{$\begin{aligned}
&\psi_a \eqbydef k' > 1 \land \textit{loop}(n, i) = \textit{loop}(i'
\times k', i'),\\
&\psi_b \eqbydef \lnot \exists k.k > 1 \land n = i \times k,\\
&\psi_i \eqbydef 2 \leq i \land \exists u.i \leq u < n \land n
\mathrel{\textit{mod}} u = 0.
\end{aligned}$}\\
The subtree
\begin{prooftree}\fontsize{9}{11}\selectfont
\AxiomC{$T_1$ \qquad $T_2$}
\RightLabel{[${\sf circ}$]}
\UnaryInfC{$\rform{\ct{\textit{loop}(n,i + 1)}{\psi_i \land \psi_b}}{\varphi_r}$}
\end{prooftree}
is:
\begin{prooftree}\fontsize{9}{11}\selectfont
\AxiomC{}
\RightLabel{[{\sf axiom}]}
\UnaryInfC{$
\begin{array}{l}
\rform{\ct{\textit{comp}}{\False}}{\varphi_r} 
\end{array}$}
\RightLabel{[{\sf subs}]}
\UnaryInfC{$\rform{\ct{\textit{comp}}{\psi_i \land \psi_b
    \land \psi_c}}{\varphi_r}$}
\AxiomC{}
\RightLabel{[{\sf axiom}]}
\UnaryInfC{$
\begin{array}{l}
\rform{\ct{\textit{loop}(n, i +
      1)}{\psi_i \land \psi_b \land \lnot \psi_c}}{\varphi_r} 
\end{array}$}
\RightLabel{[{\sf circ}]}
\BinaryInfC{$\rform{\ct{\textit{loop}(n, i + 1)}{\psi_i \land \psi_b}}{\varphi_r},$}
\end{prooftree}
\noindent where 
\\
\centerline{$\begin{aligned}
&\psi_c \eqbydef \exists n',i'.\textit{loop}(n, i+1) = \textit{loop}(n', i')
\land 2 \leq i' \land \exists u.i' \leq u < n' \land n' \mathrel{\textit{mod}} u = 0.\\
\end{aligned}$}
\\
The constraint $\psi_c$ holds when the circularity can be applied and
therefore this branch is discharged immediately by {\sf subs} and {\sf
  axiom}. The other branch, when the circularity cannot be applied, is
discharged directly by {\sf axiom}, as
$\psi_i \land \psi_b \land \lnot \psi_c$ is unsatisfiable ($\psi_i$
says that $n$ has a divisor between $i$ and $n$, $\psi_b$ says that
$i$ is not a divisor of $n$, and $\psi_c$ that $n$ has a divisor
between $i+1$ and $n$).
\\
Note that in both proof trees of the two circularities in $G$, in
order to apply the [{\sf circ}] rule, we used the following fresh
instance of the second circularity:
\\
\centerline{$\rform{\ct{\textit{loop}(n', i')}{2 \leq i' \land \exists
      u.i' \leq u < n' \land n' \mathrel{\textit{mod}} u =
      0}}{\ct{\textit{comp}}{\True}}.$}

The proof trees for both goals (circularities) in $G$ are guarded. We
have shown therefore that $(\R,G) \demEntails G$. By the Circularity
Principle (Theorem~\ref{th:cp}), we obtain that $\R \demModels G$ and
therefore
\\
\centerline{$\R \demModels 
\left\{
\begin{array}{l}
\rform{\ct{\textit{init}(n)}{\exists u.1 <
      u < n \land n \mathrel{\textit{mod}} u =
      0)}}{\ct{\textit{comp}}{\True}},\\
\rform{\ct{\textit{loop}(n, i)}{2 \leq i \land \exists u.i \leq u < n \land n \mathrel{\textit{mod}} u
= 0}}{\ct{\textit{comp}}{\True}}
\end{array}\right\}$}
\\
which includes what we wanted to show of our transition system
defined $\R$ in the running example.
\end{example}

\section{Implementation}

We have implemented the proof system for reachability in a tool called
{\tt RMT} (for {\bf r}ewriting {\bf m}odulo {\bf t}heories). {\tt RMT}
is open source and can be obtained from
\begin{center}\url{http://github.com/ciobaca/rmt/}.\end{center}

To prove a reachability property, the {\tt RMT} tool performs a
bounded search in the proof system given above. The bounds can be set
by the user. We have also tested the tool on reachability problems
where we do not use strong enough circularities. In these cases, the
tool will not find proofs. A difficulty that appears when a proof
fails, difficulty shared by all deductive approaches to correctness,
is that it is not known is the specification is wrong or if the
circularities are not strong enough. Often, analysing the failing
proof tree, the user may have the chance to find a hint for the
missing circularities, if any.  In addition, proofs might also fail
because of the incompleteness of the SMT solver. In addition to the
running example, we have used {\tt RMT} on a number of examples,
summarized in the table below:

{\small\begin{center}\begin{tabular}{|l|}
    \hline

LCTRS \hfill Reachability Property \\

\hline

Computation of $1 + \ldots + n$ \hfill Result is $n*(n+1)/2$ \\

\hline

Comp. of $\textit{gcd}{(u,v)}$ by \emph{rptd. subtractions} \hfill Result matches builtin \textit{gcd} function \\

\hline

Comp. of $\textit{gcd}{(u,v)}$ by \emph{rptd. divisions} \hfill Result matches builtin
                                       \textit{gcd} function \\

\hline

Mult. of two naturals by \emph{rptd. additions} \hfill Result matches builtin $\times$
                                 function \\

\hline

Comp. of $1^2 + \ldots + n^2$ \hfill Result is $n(n+1)(2n+1)/6$ \\

\hline

Comp. of $1^2 + \ldots + n^2$ \emph{w/out multiplications} \hfill Result is $n(n+1)(2n+1)/6$ \\

\hline

Semantics of an IMPerative language \hfill Program computing $1+\ldots+n$ is correct \\

\hline

Semantics of a FUNctional language  \hfill Program computing $1+\ldots+n$ is correct
                \\

\hline

Semantics of a FUNctional language \hfill Program computing $1^2+\ldots+n^2$ is
                             correct \\

\hline

\end{tabular}\end{center}
}

\paragraph{Implementation details. } {\tt RMT} contains roughly 5000
lines of code, including comments and blank lines. {\tt RMT} depends
only on the standard C++ libraries and it can be compiled by any
relatively modern C++ compiler out of the box. At the heart of {\tt
  RMT} is a hierarchy of classes for representing variables, function
symbols and terms. Terms are stored in DAG format, with maximum
structure sharing. The {\tt RMT} tool relies on an external SMT solver
to check satisfiability of constraints. By default, the only
dependency is the {\tt Z3} SMT solver, which should be installed and
its binary should be in the system path. A compile time switch allows
to use any other SMT solver that supports the SMTLIB interface, such
as CVC4~\cite{CVC4}. In order to reduce constraints over the full
signature to constraints over the builtin signature, {\tt RMT} uses a
unification modulo builtins algorithm (see~\cite{wollic-2018}), which
transforms any predicate $t_1 = t_2$ (where the terms $t_1, t_2$ can
possibly contain constructor symbols) into a set of builtin
constraints.

\section{Conclusion and Future Work}
\label{sec:conclusion}

We introduced a coinduction based method for proving reachability
properties of logically constrained term rewriting systems. We use a
coinductive definition of transition systems that unifies the handling
of finite and infinite executions. We propose two proof systems for
the problem above. The first one formalizes symbolic execution in
LCTRSs coinductively, with possibly infinite proof trees. This proof
system is complete, but its infinite proof trees cannot be used in
practice as proofs. In the second proof system we add to symbolic
execution a circularity proof rule, which allows to transform infinite
proof trees into finite trees. It is not always possible to find
finite proof trees, and we conjecture that establishing a given
reachability property is higher up in the arithmetic hierarchy.

We also proposed a semantics for logically constrained term rewriting
systems as transition systems over a model combining order-sorted
terms with builtin elements such as booleans, integers, etc. The
proposed semantics has the advantage of being simpler than the usual
semantics of LCTRSs defined in~\cite{DBLP:conf/frocos/KopN13}, which
requires two reduction relations (one for rewriting and one for
computing).  The approach proposed here also removes some technical
constraints such as variable inclusion of the rhs in the lhs, which is
important in modelling open systems, where the result of a transition
is non-deterministically chosen by the environment. In addition,
working in an order-sorted setting is indispensable in order to model
easily the semantics of programming
languages. %

\begin{sloppypar}In fact, proving program properties, like correctness and equivalence,
is one application of our method. A tool such as C2LCTRS
(\url{http://www.trs.cm.is.nagoya-u.ac.jp/c2lctrs/}) can be used to
convert the semantics of a C program into a LCTRS and then RMT can
prove reachability properties of the C program. Additionally, the
operational semantics of any language can be encoded as a
LCTRS~\cite{DBLP:journals/iandc/SerbanutaRM09} and then program
correctness is reducible to a particular reachability formula. But our
approach is not limited to programs, as any system that can be
modelled as a LCTRS is also amenable to our approach. We define
reachability in the sense of partial correctness (infinite execution
paths are not considered). Therefore termination should be established
in some other way~\cite{DBLP:journals/corr/Kop16}, as it is an
orthogonal concern. Our approach to reachability and LCTRSs extends to
working modulo AC (or more generally, modulo any set of equations E),
but we have not formally presented this to preserve brevity and
simplicity. For future work, we would like to test our approach on
other interesting problems that arrise in various domains. In
particular, it would be interesting to extend our approach to
reachability in the context of program
equivalence~\cite{DBLP:journals/fac/CiobacaLRR16}. An interesting
challenge is to add defined operations to the algebra underlying the
constrained term rewriting systems, which would allow a user to define
their own functions, which are not necessarily builtin.\end{sloppypar}

{\bf Acknowledgements.} We thank the anonymous reviewers for their
valuable suggestions. This work was supported by a grant of the
Romanian National Authority for Scientific Research and Innovation,
CNCS/CCCDI - UEFISICDI, project number PN-III-P2-2.1-BG-2016-0394,
within PNCDI III.

\bibliographystyle{plain}
\bibliography{refs}

\clearpage

\appendix

\section{Preliminaries}
\label{sec:prelim}

\subsection{Order-Sorted Algebra.}
\label{sec:osa}
In this subsection we recall the main definitions and notations from order-sorted algebra we use in this paper. More details can be found, e.g., in~\cite{osa1}.

An \emph{order-sorted signature} $\Sigma\eqbydef (S,\le, F)$ consists
of:
\begin{enumerate}
\item A set $S$ of \emph{sorts},
\item An $S^*\times S$-indexed family
  $F=\{F_{w, s}\mid w \in S^*, s\in S\}$ of sets whose elements are
  called \emph{operation symbols}, and
\item A partial order ${\le}\subseteq S\times S$,
\end{enumerate}
such that the following \emph{monotonicity condition} is
satisfied:%
\[f\in F_{w_1, s_1}\cap F_{w_2, s_2}\textrm{~and~}w_1\le
  w_2\textrm{~imply~}s_1\le s_2.\]
We often write $f:s_1\times\cdots\times s_n\to s$ for $f\in F_{w,s}$
with $w=s_1\ldots s_n$. We write $\varepsilon$ for the empty sequence
of sorts. A \emph{connected component} of $(S,\le)$ is an
$\simeq$-equivalence class, where $\simeq$ is the smallest equivalence
relation containing $\le$.

Given an order-sorted signature $\Sigma=(S,\le, F)$, a
\emph{$\Sigma$-model ($\Sigma$-algebra)} $M$ consists of:
\begin{itemize}
\item An $S$-indexed family $\{M_s\mid s\in S\}$ of \emph{carrier
    sets} such that $s\le s'$ implies
  $M_s\subseteq M_{s'}$;
\item A \emph{function}\footnote{If $w=s_1\ldots s_n$ then
    $M_w=M_{s_1}\times\cdots\times M_{s_n}$.} $M_f:M_w\to M_s$ for
  each operation symbol $f\in F_{w,s}$ such that if
  $f\in F_{w_1, s_1}\cap F_{w_2, s_2}$ and $w_1\le w_2$ then the
  corresponding functions $M_f: M_{w_1}\to M_{s_1}$ and
  $M_f: M_{w_2}\to M_{s_2}$ agree on $M_{w_1}$.
\end{itemize}
Let $\Sigma\eqbydef (S,\le, F)$ be an order-sorted signature and let
$X\eqbydef\{X_s\mid s\in S\}$ be an $S$-indexed family of
\emph{variables} such that $s\not=s'$ implies
$X_s\cap X_{s'}=\emptyset$. 
The $S$-indexed family
$T_\Sigma(X)=\{T_{\Sigma,s}(X)\mid s\in S\}$ of \emph{$\Sigma$-terms
  with variables $X$} is inductively defined as follows:
\begin{itemize}
\item $X_s\subseteq T_{\Sigma,s}(X)$;
\item if $f\in F_{w,s}$, $w=s_1\ldots s_n$, and
  $t_i\in T_{\Sigma,s_i}(X)$ for $i=1,\ldots, n$, then the expression
  $f(t_1,\ldots,t_n)$ belongs to $T_{\Sigma,s}(X)$;
\item if $s\le s'$ then $T_{\Sigma,s}(X)\subseteq T_{\Sigma,s'}(X)$.
\end{itemize}

We also make the standard assumption that constant symbols are not
followed by parentheses.

$T_\Sigma(X)$ can be organised as a $\Sigma$-model by considering
$T_\Sigma(X)_f:T_\Sigma(X)_w\to T_\Sigma(X)_s$ that maps the terms
$t_1,\ldots,t_n$ into $f(t_1,\ldots,t_n)$, where $f\in F_{w,s}$. The
set of \emph{ground terms} is $T_\Sigma = T_\Sigma(\emptyset)$. In
this paper we consider order-sorted signatures $(S,\le, F)$ that are
\emph{preregular}~\cite{osa1}, i.e., each term $t$ in $T_\Sigma(X)$
has a least sort $\ls(t)$. For this case, $T_\Sigma$ is a initial
order-sorted $\Sigma$-algebra and $T_\Sigma(X)$ is a free order-sorted
$\Sigma$-algebra. A \emph{context} $c[\cdot]$ is a term including
exactly one occurrence of a distinguished variable $\cdot$, and $c[t]$
denotes the term obtained from $c[\cdot]$ by replacing the $\cdot$
variable with the term $t$. The set of variables occurring in a term
$t$ is denoted by $\var(t)$, and $\var(t_1,\ldots,t_n)$ denotes
$\var(t_1)\cup\cdots\cup \var(t_n)$.

Given a $\Sigma$-model $M$, a \emph{variable assignment (valuation)}
is a function $\alpha:X\to M$ that sends a variable $x\in X_s$ into a
model element $\alpha(x)\in M_s$. A valuation $\alpha : X \to M$ is
extended to terms $\alpha:T_\Sigma(X)\to M$ by setting
$\alpha(f(t_1,\ldots,t_n))=M_f(\alpha(t_1),\ldots,\alpha(t_n))$. A
\emph{substitution} is a valuation $\sigma : X \to T_\Sigma(X)$ such
that the domain of $\sigma$,
$\dom(\sigma) = \{x \in X\mid \sigma(x)\not= x\}$, is
finite. Functions $\sigma : Y \to T_\Sigma(X)$ defined only on a
subset $Y \subset X$ and having a finite domain $\dom(\sigma)$ are
identified with the unique substitution $\sigma^e : X \to T_\Sigma(X)$
with $\dom(\sigma^e) = \dom(\sigma)\subseteq Y$. The identity substitution
$\textit{id} : X \to T_\Sigma(X)$ is defined as the unique
substitution with $\dom(\textit{id}) = \emptyset$.

If the partial order ${\le} \subseteq S \times S$ is the equality
relation, i.e. $s \le s'$ iff $s = s'$, then
$\Sigma \eqbydef (S, \le, F)$ is a \emph{many-sorted signature} and we
simply write $\Sigma \eqbydef (S, F)$.

An order-sorted signature $(S',\le', F')$ is a \emph{closed
  subsignature} of $(S,\le, F)$ if: 
\begin{enumerate}
\item $S'\subseteq S'$ and ${\le'}={\le}|_{S'}$,
\item $F'\subseteq F$ (as $S'\times {S'}^*$-indexed families) and
\item $f \in F_{w,s}$ and $w\in {S'}^*$ imply $f\in F'$ (and hence
  $s\in S'$).
\end{enumerate}
We refer~\cite{osa1} for a detailed presentation of order-sorted algebra.

\subsection{Coinduction.}
\label{sec:coind}
We briefly recall from~\cite{SangiorgiBook} (Chapter 2) the
coinductive definitions and the coinduction proof technique defined
using inference rules.

Given a set $U$, a \emph{ground inference rule} over $U$ is a tuple
$(a_1, \ldots, a_n, a)$\footnote{Here we consider only case when the set of the premises is finite.}, 
often written \[\dfrac{a_1 \ldots a_n}{a},\]
where $a, a_1, \ldots, a_n \in U$.

Given a set $R$ of inference rules over $U$ and a subset
$X \subseteq U$, the \emph{one-step closure of $X$ with respect to
  $R$} is the set $\widehat{R}(X) \subseteq U$ defined as
follows:
\[\widehat{R}(X) \eqbydef \left\{ a \,\middle |\, \dfrac{a_1 \ldots a_n}{a} \in R,
  a_1, \ldots, a_n \in X \right\}.\]

$\widehat{R}$ is a monotone endofunction defined over ${\cal P}(U)$,
which is a complete lattice. Hence $\widehat{R}$ has a least fixed
point and a greatest fixed point, by the Fixed Point Theorem
(see~\cite{SangiorgiBook}). By $\nu \widehat{R}$ we denote
\emph{the greatest fixed point of $\widehat{R}$}, that is the largest
set $X \subseteq U$ such that $X \supseteq \widehat{R}(X)$
(equivalently, $X$ is the largest set such that $X =
\widehat{R}(X)$). We say that a set is \emph{coinductively defined} if
it is the greatest fixed point $\nu \widehat{R}$ of some ground rule
system $R$.

\begin{example}
\label{ex:lists}
The system LIST, given below,  coinductively defines the possibly infinite lists over integers:
\[
[A]~\dfrac{}{\mathit{nil}}\qquad [B]~\dfrac{\ell}{z,\ell}~z\in\mathbb{Z}
\]
Let $\U$ be any set including all (finite and infinite) strings over $\mathbb{Z}\cup\{\mathit{nil},{\,}{,}{\,}\}$.  
Note that $[B]$ is a rule scheme; the ground rules are obtained by instantiating $z$ and $\ell$ with concrete integers and elements in $\U$, respectively. 
We have\\
\centerline{$\widehat{\sf LIST}(X)=\{z,\ell\mid z\in \mathbb{Z},\ell\in X\}\cup\{\mathit{nil}\}$.}\\
The set of possibly infinite (i.e., finite and infinite) lists is the greatest fixed point, $\mathbb{Z}^\infty=\nu\, \widehat{\sf LIST}$.
We obtain\\
\centerline{$
\begin{aligned}
\mathbb{Z}^\infty &=\U\cap\widehat{\sf LIST}(\U)\cap \widehat{\sf LIST}^2(\U)\cap\ldots\\
&= \U\cap (\{z_1,u\mid z_1\in \mathbb{Z},u\in \U\}\cup\{\mathit{nil}\})\\
&\phantom{{}= \U}\cap(\{z_1, z_2,u\mid z_1,z_2\in \mathbb{Z},u\in \U\}\cup\{z_1,\mathit{nil}\mid z_1\in \mathbb{Z}\}\cup\{\mathit{nil}\})\\
&\phantom{{}= \U}\cap\ldots
\end{aligned}$}\\ 
by Kleene's Theorem.
The set of infinite lists over $\mathbb{Z}$ is the greatest fixed point of the system consisting only of rule $[B]$, i.e. $\mathbb{Z}^\omega=\nu\,\widehat{[B]}$. 
\end{example}

  To prove that some element $x$ is in $\nu \widehat{R}$, we often use the
  well-known \emph{coinduction principle}:

  \begin{proposition}[The Coinduction Principle] Let $X \subseteq U$
    be a set such that $X \subseteq \widehat{R}(X)$. If $x \in X$,
    then $x \in \nu \widehat{R}$.

  \end{proposition}
The coinduction principle can be represented in a more compact way by the following inference rule:
\[
\dfrac{X \subseteq \widehat{R}(X)}{X \subseteq \nu \widehat{R}}
\]
\begin{definition}
\label{def:coadm}
A rule $r$ is called \emph{coadmissible} for $R$ if  $\nu\,\widehat{R}=\nu\,\widehat{R\cup\{r\}}$.
\end{definition}
\begin{remark}
Adding a coadmissible rule $R$ does not change the greatest fixed point. Usually, a coadmissible rule can be used only finitely many times in a proof tree. Otherwise we may have unsound proofs, like that consisting only of coadmissible rules.
The main idea behind of a coadmissible rule $r$ is that if its premises can be derived using rules from $R$ (i.e. there are proof trees under $R$ for its premises), then we may find a proof tree under $R$ of its conclusion. Then, we can show by induction that any proof tree under $R\cup\{r\}$, where $r$ is applied only finitely many times, can be transformed into a proof tree under $R$ for the same conclusion.
\end{remark}

  We make extensive use of sets coinductively defined by rules. The
  underlying set $U$ will be understood each time from the shape of
  rules: e.g., if the hypotheses and the conclusion are pairs of
  execution paths and formulas, then the set $U$ is the set of all
  these pairs.

\section{Proofs of Helper Results}

The following result is a direct consequence of Definition~\ref{def:dvp}.
\begin{corollary}

  If $(M,\gtran{})\demModels\rform{P}{Q}$ then
  $(M,\gtran{})\demModels\rform{\partial(P\setminus Q)}{Q}$.

\end{corollary}

The disjunction of valid predicates with the same target is a valid
predicate as well:

\begin{proposition} \label{prop:joinrp} If $(M,\gtran{})\demModels\rform{P_i}{Q}$ for $i=1,2$, then
$(M,\gtran{})\demModels\rform{P_1\cup P_2}{Q}$.
\end{proposition}
\begin{proof} 
We show that the set $X=\{\rform{P_1\cup P_2}{Q}\mid (M,\gtran{})\demModels\rform{P_i}{Q}, i=1,2\}$ is backward closed w.r.t. $\widehat{\sf DVP}$, i.e. $X\subseteq \widehat{\sf DVP}(X)$. Note first that if $(M,\gtran{})\demModels\rform{P}{Q}$ then obviously $\rform{P}{Q}\in X$ since $P=P\cup\emptyset$ and $(M,\gtran{})\demModels\rform{\emptyset}{Q}$.
\\
Let $\rform{P_1\cup P_2}{Q}\in X$. 
We have $\partial((P_1\cup P_2)\setminus Q)=\partial(P_1\setminus Q)\cup \partial(P_2\setminus Q)$, which implies $\rform{\partial((P_1\cup P_2)\setminus Q)}{Q}\in X$ (since $(M,\gtran{})\demModels\rform{\partial(P_i\setminus Q)}{Q}$, $i=1,2$).
It follows that $\rform{P_1\cup P_2}{Q}\in \widehat{\sf DVP}(X)$ by the rule $\semrule{Step}$.
\qed\end{proof}
\begin{corollary}\label{cor:joinrp}
$X = \nu\,\widehat{\sf DVP}$, where $X$ is the set from the proof of Proposition~\ref{prop:joinrp}.
\end{corollary}

Proposition~\ref{prop:joinrp} allows to extend {\sf DVP} with the
following coadmissible inference rule: 
  \[
  {\semrule{Union}}~\dfrac{\rform{P_1}{Q}~\rform{P_2}{Q}}{\rform{P_1\cup
      P_2}{Q}}~P_1\not=\emptyset\not=P_2. \]

The following result shows that the set of demonic valid reachability predicates is closed under the subset relation.

\begin{proposition} \label{prop:subsetrp} If
  $(M,\gtran{})\demModels\rform{P}{Q}$ and $P'\subseteq P$ then
  $(M,\gtran{})\demModels\rform{P'}{Q}$.
\end{proposition}
\begin{proof} 
We show that the set $X=\{\rform{P'}{Q}\mid P'\subseteq P, (M,\gtran{})\demModels\rform{P}{Q}\}$ is backward closed w.r.t. $\widehat{\sf DVP}$, i.e. $X\subseteq \widehat{\sf DVP}(X)$.
\\
Let $\rform{P'}{Q}\in X$. It follows that there is $P$ such that $(M,\gtran{})\demModels\rform{P}{Q}$ and $P'\subseteq P$. Let $\it PT$ be a proof tree of $\rform{P}{Q}$ under $\widehat{\sf DVP}$. We distinguish the following two cases:
\\
1. $P\subseteq Q$. It follows that $P'\subseteq Q$ and hence $\rform{P'}{Q}\in \widehat{\sf DVP}(X)$.
\\
2. The unique child of the root is $\rform{\partial(P\setminus Q)}{Q}$. We have $\partial(P'\setminus Q)\subseteq \partial(P\setminus Q)$ and hence $\rform{\partial(P'\setminus Q)}{Q}\in X$, which implies $\rform{P'}{Q}\in \widehat{\sf DVP}(X)$.
\qed\end{proof}

To show the demonic validity of a reachability predicate, we have to
find a proof tree only for the state predicate of not already reached
target states:

\begin{proposition} \label{prop:redrp} If
  $(M,\gtran{})\demModels\rform{P}{Q}$ iff
  $(M,\gtran{})\demModels\rform{P\setminus Q}{Q}$.
\end{proposition}
\begin{proof}
\textit{Reverse implication} ($\Leftarrow$). We have $P=(P\setminus Q)\cup (P\cap Q)$.  Since $P\cap Q\subseteq Q$, we obviously have $(M,\gtran{})\demModels\rform{P\cap Q}{Q}$. The conclusion follows for by applying Proposition~\ref{prop:joinrp}.
\\
\textit{Direct implication} ($\Rightarrow$). Since $(P\setminus Q)\subseteq P$, the conclusion follows by Proposition~\ref{prop:subsetrp}.
\qed\end{proof}

A starting state of a demonically valid reachability predicate that is not in the target state predicate must be runnable:

\begin{proposition} \label{prop:runnablerp} If
  $(M,\gtran{})\demModels\rform{P}{Q}$ then $P\setminus Q$ is
  runnable.
\end{proposition}
\begin{proof}
It follows directly from the definition of {\sf DVP}.
\qed\end{proof}
\begin{corollary}
\label{cor:runnablerp}
If $P\cap Q=\emptyset$ and $(M,\gtran{})\demModels\rform{P}{Q}$ then $P$ is runnable.
\end{corollary}

\begin{proposition} \label{pro:eqrform1} If
  $\var(\varphi_1)\cap\var(\varphi')=\var(\varphi_2)\cap\var(\varphi')$
  and $\tsem{\sigma(\varphi_1)}= \tsem{\sigma(\varphi_2)}$ for all
  $\sigma:\var(\varphi_i)\cap\var(\varphi')\to M^\Sigma$, then
  $\rform{\varphi_1}{\varphi'}\equiv_\R\rform{\varphi_2}{\varphi'}$.
\end{proposition}
\begin{proof}
$\rform{\varphi_1}{\varphi'}$ and $\rform{\varphi_2}{\varphi'}$ define the same reachability predicate for each $\sigma:\var(\varphi_i)\cap\var(\varphi')\to M^\Sigma$.
\qed\end{proof}
\begin{proposition} \label{pro:eqrform2}
  $\rform{\ct{t}{\phi}}{\varphi'}\equiv
  \rform{\ct{z}{z \qeq t\land\phi}}{\varphi'}$, where $z$ is a fresh
  variable (it does not appear in $\rform{\ct{t}{\phi}}{\varphi'}$).
\end{proposition}
\begin{proof}
We obviously have $\var(\ct{t}{\phi})\cap\var(\varphi')=\var(\ct{z}{z \qeq t\land\phi})\cap\var(\varphi')$ and $\tsem{\sigma(\ct{t}{\phi})}= \tsem{\sigma(\ct{z}{z \qeq t\land\phi})}$  for all
  $\sigma:\var(\varphi_i)\cap\var(\varphi')\to M^\Sigma$ (we used here the fact that $\sigma(z)=z$). 
  Then we apply Proposition~\ref{pro:eqrform1}.
\qed\end{proof}
\begin{proposition} \label{pro:eqrform3} If
  $\rform{\ct{t}{\phi_i}}{\varphi'}\equiv\rform{\ct{t''_i}{\phi''_i}}{\varphi'}$ and
  $\var(\ct{t}{\phi_i})\cap\var(\varphi')=\var(\ct{t''_i}{\phi''_i})\cap\var(\varphi')$
  for $i=1,2$, then
  \\
  \indent{$\rform{\ct{t}{\phi_1\lor\phi_2}}{\varphi'}$}
  \\
  is equivalent to
  \\
  \indent{$\rform{\ct{z}{(z \qeq t''_1\land\phi''_1)\lor
        (z \qeq t''_2\land\phi''_2)}}{\varphi'}$}
  \\
  where $z$ is a fresh variable.
\end{proposition}
\begin{proof}
We have $\rform{\ct{t''_i}{\phi''_i}}{\varphi'}$ equivalent to $\rform{\ct{z}{z \qeq t''_i\land\phi''_i}}{\varphi'}$ by Proposition~\ref{pro:eqrform2}.
Let $\R$ be a constrained rule system, 
$\sigma:\var(\ct{t}{\phi_1\lor\phi_2})\cap\var(\varphi')\to M^\Sigma$ and assume that 
\\%
\indent{$(M^\Sigma,\gtran{\R})\demModels\rform{\tsem{\sigma(\ct{t}{\phi_1\lor\phi_2})}}{\tsem{\sigma(\varphi')}}$.}
\\%
Note that $\var(\ct{t}{\phi_1\lor\phi_2})\cap\var(\varphi')=\var(\ct{z}{(z \qeq t''_1\land\phi''_1)\lor(z \qeq t''_1\land\phi''_1)}$ by the hypotheses.
Since $\tsem{\sigma(\ct{t}{\phi_i})}\subseteq\tsem{\sigma(\ct{t}{\phi_1\lor\phi_2})}$, it follows that
\\%
\indent{$(M^\Sigma,\gtran{\R})\demModels\rform{\tsem{\sigma(\ct{t}{\phi_i})}}{\tsem{\sigma(\varphi')}}$}
\\%
by Proposition~\ref{prop:subsetrp}, $i=1,2$.
By Proposition~\ref{pro:eqrform2} we obtain 
\\%
\indent{$(M^\Sigma,\gtran{\R})\demModels\rform{\tsem{\sigma(\ct{z}{z \qeq t''_i\land\phi''_i})}}{\tsem{\sigma(\varphi')}}$}
\\%
for $i=1,2$, which implies
\\%
\indent{$(M^\Sigma,\gtran{\R})\demModels\rform{\tsem{\sigma(\ct{z}{z \qeq t''_1\land\phi''_1})}\cup\tsem{\sigma(\ct{z}{z \qeq t''_1\land\phi''_1})}}{\tsem{\sigma(\varphi')}}$}
\\%
by Proposition~\ref{prop:joinrp}. Since
\\%
\indent{$\begin{array}{l}
\tsem{\sigma(\ct{z}{z \qeq t''_1\land\phi''_1})}\cup\tsem{\sigma(\ct{z}{z \qeq t''_2\land\phi''_2})}
\\
=
\\
\tsem{\sigma(\ct{z}{(z \qeq t''_1\land\phi''_1)\lor (z \qeq t''_2\land\phi''_2)})}
\end{array}$}
\\%
it follows that
\\%
\indent{$(M^\Sigma,\gtran{\R})\demModels\rform{\tsem{\sigma(\ct{z}{(z \qeq t''_1\land\phi''_1)\lor (z \qeq t''_2\land\phi''_2)})}}{\tsem{\sigma(\varphi')}}$.}
\\%
Since $\sigma$ defined over $\var(\ct{t}{\phi_1\lor\phi_2})\cap\var(\varphi')$
is arbitrary, we have proved that
\\%
\indent$\R\demModels \rform{\ct{t}{\phi_1\lor\phi_2}}{\varphi'}$ implies $\R\demModels \rform{\ct{z}{(z \qeq t''_1\land\phi''_1)\lor (z \qeq t''_2\land\phi''_2)}}{\varphi'}$.
\\%
The converse implication is proven in a similar way.
\qed\end{proof}

 \begin{remark}
 \label{rem:eqrf}
 We assume that $\var(\varphi_1)\cap\var(\varphi')=\var(\varphi_2)\cap\var(\varphi')$ and $\tsem{\varphi_1}=\tsem{\varphi_2}$ whenever
 $\rform{\varphi_1}{\varphi'}\equiv_\R\rform{\varphi_2}{\varphi'}$. 
 The first equality says that the shared variables by the lhs and rhs are preserved by the equivalence and the second one is needed to be sure that the two constrained terms have the same syntactic derivatives (see below).
 \end{remark}

The following result is useful for case analysis:

\begin{restatable}{proposition}{disjprop} \label{prop:disj} If
 $M^\Sigma\models \phi\iff(\phi_1\lor\phi_2)$,
  $\rform{\ct{t}{\phi_1}}{\varphi'}$ and
  $\rform{\ct{t}{\phi_2}}{\varphi'}$ are in
  $\nu\,\widehat{{\sf DSTEP}(\R)}$, and
  $\var(\ct{t}{\phi_1})\cap\var(\varphi')=\var(\ct{t}{\phi_2})\cap\var(\varphi')$,
  then $\rform{\ct{t}{\phi}}{\varphi'}$ is in
  $\nu\,\widehat{{\sf DSTEP}(\R)}$.
\end{restatable}

\noindent and it allows to extend ${\sf DSTEP}(\R)$ with the
following inference rule%
:
\begin{definition}[Coadmissible rule for reachability formulae]
  \[ [{\sf disj}]~\dfrac{\rform{\ct{t}{\phi_1}}{\varphi'},
      \rform{\ct{t}{\phi_2}}{\varphi'}}{\rform{\ct{t}{\phi}}{\varphi'}}~M^\Sigma\models
    \phi\iff \phi_1\lor\phi_2
\]
\end{definition}
\begin{proof}[of Proposition~\ref{prop:disj}]
Let $A$ be the set
\\%
\centerline{$\left\{\!\rform{\ct{t}{\phi}}{\varphi'}\,\middle|\, \rform{\ct{t}{\phi_1}\!}{\!\varphi'}, \rform{\ct{t}{\phi_2}\!}{\varphi'\!}\in\nu\,\widehat{{\sf DSTEP}(\R)},M^\Sigma\models \phi\iff(\phi_1\lor\phi_2)\!\right\}$.}
\\%
Note that $\nu\,\widehat{{\sf DSTEP}(\R)}\subseteq A$ since $\phi$ is equivalent to $\phi\lor \phi$, which implies
$\nu\,\widehat{{\sf DSTEP}(\R)}\subseteq \widehat{{\sf DSTEP}(\R)}(A)$ ($\clubsuit$). 
We show that $A$ is backward-closed w.r.t. ${\sf DSTEP}(\R)$, i.e. $A\subseteq\widehat{{\sf DSTEP}(\R)}(A)$.
Let $\rform{\ct{t}{\phi}}{\varphi'}\in A$, where $M^\Sigma\models\allowbreak \phi\iff(\phi_1\lor\phi_2)$.
Let ${\it PT}_i$ a proof tree for $\rform{\ct{t}{\phi_i}}{\varphi'}$ under ${\sf DSTEP}(\R)$, $i=1,2$.
We distinguish the following cases, according to the definition of ${\it PT}_i$, $i=1,2$:
\\
1. $M^\Sigma\models\phi_i\iff \False$, $i\in\{1,2\}$ (${\it PT}_i$ consists of [{\sf axiom}]). Then  $\rform{\ct{t}{\phi_1{\lor}\phi_2}}{\varphi'}$ is equi\-valent to $\rform{\ct{t}{\phi_{3-i}}\!}{\!\varphi'}$, which is in
$\nu\,\widehat{{\sf DSTEP}(\R)}$ and hence in $\widehat{{\sf DSTEP}(\R)}(A)$ by ($\clubsuit$).  
\\
2. The rule corresponding to the root of ${\it PT}_i$ is [{\sf subs}] for $i=1,2$. 
By Proposition~\ref{pro:eqrform2}, we may assume that the child of the ${\it PT}_i$'root is of the form
\\%
\indent{$\rform{\ct{z}{z \qeq t''_i\land\phi''_i\land\neg\phi'''}}{\varphi'}$,}
\\%
$i=1,2$, where $\phi'''$ is of the form $(\exists \widetilde{x})z{ \qeq }t'\land\phi'$, and $\widetilde{x}=\var(\varphi')\setminus\var(\ct{t''_i}{\phi''_i}){=}\allowbreak\var(\varphi')\setminus\var(\ct{t}{\phi_i})$, $i=1,2$ (by the hypotheses of the proposition and Remark~\ref{rem:eqrf}).
Since $\rform{\ct{t}{\phi_i}}{\varphi'}$ is equivalent to $\rform{\ct{z}{z \qeq t''_i\land\phi''_i}}{\varphi'}$, it follows that
$\rform{\ct{t}{\phi_1\lor\phi_2}}{\varphi'}$ is equivalent to $\rform{\ct{z}{(z \qeq t''_1\land\phi''_1)\lor (z \qeq t''_2\land\phi''_2)}}{\varphi'}$ by Proposition~\ref{pro:eqrform3}. 
It follows that 
\\%
\indent{$\rform{\ct{z}{((z \qeq t''_1\land\phi''_1)\lor (z \qeq t''_2\land\phi''_2))\land\neg\phi'''}}{\varphi'}$}
\\%
is equivalent to
\\%
\indent{$\rform{\ct{z}{(z \qeq t''_1\land\phi''_1\land\neg\phi''')\lor (z \qeq t''_2\land\phi''_2\land\neg\phi''')}}{\varphi'}$}
\\%
which is in $A$ and hence $\rform{\ct{t}{\phi_1\lor\phi_2}}{\varphi'}\in \widehat{{\sf DSTEP}(\R)}(A)$ by [{\sf subs}].
\\
3. The rule corresponding to the root of ${\it PT}_i$ is [${\sf der}^\forall$] for $i=1,2$. We assume that the children of the root of ${\it PT}_i$ are of the form $\rform{\ct{{t''}^j}{{\phi''}^j_i}}{\varphi'}$ with 
\\%
\indent{$\ct{{t''}^j}{{\phi''}^j_i}\in\Delta_\R(\ct{t''_i}{\phi''_i})$}
\\%
and 
\\%
\indent{$\rform{\ct{t}{\phi_i}}{\varphi'}\equiv\rform{\ct{t''_i}{\phi''_i}}{\varphi'}$}
\\%
where $j\in J''_i$ and $i=1,2$. We also assume that 
\\%
\indent{$\Delta_\R(\ct{t}{\phi_i})=\{\ct{{t}^j}{{\phi}^j_i}\mid j\in J_i\}$}
\\%
where $i=1,2$. For $j\in J_i$ there is $j''\in J''_i$ such that
\\%
\indent{$\rform{\ct{t^j}{\phi^j_i}}{\varphi'}\equiv\rform{\ct{{t''}^{j''}}{{\phi''}^{j''}_i}}{\varphi'}$}
\\%
by Remark~\ref{rem:eqder}.
Since $\rform{\ct{{t''}^{j''}}{{\phi''}^{j''}_i}\!}{\!\varphi'}\in \nu\,\widehat{{\sf DSTEP}(\R)}$, it follows that 
$\rform{\ct{t^j}{\phi^j_i}\!}{\varphi'}\in \nu\,\widehat{{\sf DSTEP}(\R)}$ by the implicit equivalence rule, $i=1,2$.
Note that $\phi^j_i$ is of the form (or equivalent to) $\phi_i\land \phi^j$, where $\phi^j$ depends only on $t$ and the applied rule,
which implies
\\%
\indent{$\Delta_\R(\ct{t}{\phi_1\lor\phi_2})=\{\ct{{t}^j}{(\phi_1\lor\phi_2)\land{\phi}^j}\mid j\in J_1\cup J_2\}$}
\\%
Using the equivalence between $(\phi_1\lor\phi_2)\land{\phi}^j$ and $(\phi_1\land\phi^j)\lor(\phi_2\land{\phi}^j)$, we obviously obtain
$\rform{\ct{t^j}{(\phi_1\lor \phi_{2})\land\phi^j}}{\varphi'}\in A$, and hence 
\\%
\indent{$\rform{\ct{t}{\phi_1\lor\phi_{2}}}{\varphi'}\in \nu\,\widehat{{\sf DSTEP}(\R)}(A)$}
\\%
by [${\sf der}^\forall$].
\\
4. The rule corresponding to the root of ${\it PT}_i$ is [{\sf subs}] and the rule corresponding to the root of ${\it PT}_{3-i}$ is [${\sf der}^\forall$], $i\in\{1,2\}$. Note that the rule [{\sf subs}] cannot be applied twice consecutively, so the child of the root of ${\it PT}_i$ corresponds to either [{\sf axiom}] or [${\sf der}^\forall$]. The rest of the proof for this case is similar to the case 1 or to the case 3.
\qed\end{proof}

The next result shows that it is fine to relax the constraints of some goals (a kind of generalization).

\begin{proposition} \label{prop:dcc-impl} If
  $(\R,G)\demEntails \rform{\ct{t}{\phi}}{\varphi'}$ then
  $(\R,G)\demEntails \rform{\ct{t}{\phi\land\phi''}}{\varphi'}$, where
  $\phi''$ is a constraint formula.
\end{proposition}
\begin{proof}
Let $\it PT$ a guarded proof tree for $\rform{\ct{t}{\phi}}{\varphi'}$ under ${\sf DCC}(\R,G)$.
We transform $\it PT$ into a guarded proof tree $\it PT'$ for $\rform{\ct{t}{\phi\land\phi''}}{\varphi'}$ under ${\sf DCC}(\R,G)$ as follows:
The root $\rform{\ct{t}{\phi}}{\varphi'}$ is transformed into $\rform{\ct{t}{\phi\land\phi''}}{\varphi'}$. Assuming that the current node $\rform{\ct{t_1}{\phi_1}}{\varphi'}$ is transformed into $\rform{\ct{t'_1}{\phi'_1}}{\varphi'}$ with $M^\Sigma\models \phi'_1\iff \phi_1\land\phi''$, its children are transformed according to the inference rule used to obtain the current node (if the rule involves an equivalence of the conclusion, then $\rform{\ct{t_1}{\phi_1}}{\varphi'}$ is the used equivalent formula):
\begin{enumerate}
\item {} [{\sf axiom}]. $M^\Sigma\models \phi_1\iff \False$ implies $M^\Sigma\models \phi'_1\iff \False$ and there are no children in this case.
\item {}[{\sf subs}]. The unique child of $\rform{\ct{t_1}{\phi_1}}{\varphi'}$ is of the form $\rform{\ct{t_1}{\phi''_1}}{\varphi'}$ with $M^\Sigma\models \phi''_1\iff \phi_1\land\phi_2$ and it is transformed into  $\rform{\ct{t_1}{\phi''_1\land\phi''}}{\varphi'}$. 
Note that $\phi_2$ does not depend on $\phi_1$, so it is the same with that for $\rform{\ct{t_1}{\phi_1\land\phi''}}{\varphi'}$.
Obviously, $\rform{\ct{t'_1}{\phi'_1}}{\varphi'}$ and $\rform{\ct{t_1}{\phi''_1\land\phi''}}{\varphi'}$ form an instance of [{\sf subs}].
\item {}[${\sf der}^\forall$]. The children of the current node in $\it PT$ are of the form $\rform{\ct{t^j_1}{\phi^j_1}}{\varphi'}$, $j\in J$. We have $\Delta_\R(\ct{t_1}{\phi'_1})=\{\ct{t^j_1}{{\phi'}^j_1}\mid j\in J'\}$ with $M^\Sigma\models {\phi'}^j_1\iff \phi^j_1\land\phi''$ and $J'\subseteq J$ by the definition of $\Delta_\R$. We may have $J'\subset J$ because some of $\phi^j_1\land\phi''$ could become unsatisfiable. The children of $\rform{\ct{t'_1}{\phi'_1}}{\varphi'}$ are $\rform{\ct{t^j_1}{{\phi'}^j_1}}{\varphi'}$, $j\in J'$. Obviously, the new node is an instance of [${\sf der}^\forall$].
\item {}[{\sf circ}]. The curent node in $\it PT$ has two children of the form 
\\
\indent{$\rform{\ct{t'_c}{\phi'_c\land \phi_1\land \phi''_1}}{\varphi'}$ and $\rform{\ct{t_1}{\phi_1\land\neg\phi''_1}}{\varphi'}$.}
\\
The children are transformed into 
\\
\indent{$\rform{\ct{t'_c}{\phi'_c\land \phi_1\land\phi''\land \phi''_1}}{\varphi'}$ and $\rform{\ct{t_1}{\phi_1\land\phi''\land\neg\phi''_1}}{\varphi'}$,}
\\ 
respectively.
\end{enumerate}
Since $\it PT$ is guarded it follows that $\it PT'$ is guarded as well.
\qed\end{proof}

\section{Proofs of Results from the Paper}

\demonicvalidsoundness*
\begin{proof}
\textit{Reverse implication} ($\Leftarrow$). Let $X$ denote the set $\{\langle \tau, \rform{P}{Q}\rangle\mid \hd(\tau)\in P, {(M,\gtran{})}\allowbreak\demModels\rform{P}{Q}\}$.
We show that $X$ is backward closed w.r.t. $\widehat{\sf EPSRP}$, i.e. $X\subseteq \widehat{\sf EPSRP}(X)$. 
Let $\langle \tau, \rform{P}{Q}\rangle$ be in $X$.
If $\hd(\tau)\in Q$ then $\langle \tau, \rform{P}{Q}\rangle$ is in $\widehat{\sf EPSRP}(X)$ by the first rule of ${\sf EPSRP}$.
If $\hd(\tau)\in P\setminus Q$ then $(M,\gtran{})\demModels\rform{\partial(P\setminus Q)}{Q}$ and $P\setminus Q$ is runnable (the root of the proof tree for $\rform{P}{Q}$ corresponds to the second rule in {\sf DVP}).  It follows that $\tau$ is of the form $\gamma_0\gtran{}\tau'$. We have $\hd(\tau')\in \partial(P\setminus Q)$ by the definition of $\partial(\_)$. It follows that $\langle \tau', \rform{\partial(P\setminus Q)}{Q}\rangle\in X$, which implies 
$\langle \tau, \rform{P}{Q}\rangle\in\widehat{\sf EPSRP}(X)$ by the second rule of ${\sf EPSRP}$. 
This finishes the proof of "if" direction.
\\
\textit{Direct implication} ($\Rightarrow$). Let $Y$ be the set $\{\rform{P'}{Q}\mid (\forall\tau)\hd(\tau)\in P'\limplies \tau\demModels \rform{P'}{Q}\}$.
We show that $Y$ is backward closed w.r.t. $\widehat{\sf DVP}$, i.e. $Y\subseteq\widehat{\sf DVP}(Y)$, and we then apply the coinduction rule.
Let $\rform{P'}{Q}\in Y$. 
For any $\gamma\in P'\setminus Q$ there is an execution path $\tau$ starting from $\gamma$. Since $\tau\demModels \rform{P'}{Q}$ and $\gamma\not\in Q$ it follows that $\tau=\gamma\circ\tau'$ for certain $\tau'$ with $\tau'\demModels\rform{\partial(P')}{Q}$ (by the definition of $\demModels$), which implies $P'\setminus Q$ runnable. 
Moreover, we observe that $\tau\demModels \rform{P'}{Q}$ implies $\tau\demModels \rform{P'\setminus Q}{Q}$ for any $\tau$ starting from $P'\setminus Q$, i.e. $\rform{(P'\setminus Q)}{Q}\in Y$.
We show now that $\rform{\partial(P'\setminus Q)}{Q}\in Y$.
Let $\tau'$ be an execution path with $\hd(\tau')\eqbydef \gamma'\in \partial(P'\setminus Q)$. There is $\gamma\in P'\setminus Q$ such that $\gamma\gtran{}\gamma'$ by the definition of $\partial$. Then $\tau\demModels \rform{P'\setminus Q}{Q}$ by the definition of $Y$, where $\tau=\gamma\circ\tau'$, which implies 
$\tau'\demModels \rform{\partial(P'\setminus Q)}{Q}$ by the definition of $\demModels$. Since $\tau'$ is arbitrary, it follows that $\rform{\partial(P'\setminus Q)}{Q}\in Y$, and hence $\rform{P'}{Q}\in \widehat{\sf DVP}(Y)$ by the rule $\semrule{Step}$. Since $\rform{P'}{Q}\in Y$ is arbitrary, it follows that
$Y\subseteq\widehat{\sf DVP}(Y)$. This finishes the proof of the "only if" direction.
\qed\end{proof}

\semanticinclusion*
\begin{proof}
\textit{Reverse implication} ($\Leftarrow$).
Assume that $M^\Sigma\models \phi\limplies (\exists \widetilde{x})(t \qeq t'\land \phi')$ and consider $\sigma:\var(t,\phi)\cap\var(t',\phi')\to M^\Sigma$. 
We have to prove that $\tsem{\sigma(\ct{t}{\phi})}\subseteq \tsem{\sigma(\ct{t'}{\phi'})}$.
Let $\alpha$ be a valuation such that $M^\Sigma,\alpha\models\sigma(\phi)$. Then the valuation $\alpha_1$, given by
$\alpha_1(y)=\textrm{if }y\in\var(t,\phi)\cap\var(t',\phi')\textrm{ then }\sigma(y)\textrm{ else }\alpha(y)$, satisfies $M^\Sigma,\alpha_1\models\phi$. There is a valuation $\alpha'_1$ such that $\alpha'_1(y)=\alpha_1(y)$, for each $y\not\in\widetilde{x}$, and  $M^\Sigma,\alpha'_1\models (t \qeq t'\land \phi')$.  Since $\var(t,\phi)\cap\var(t',\phi')$ and $\widetilde{x}$ are disjoint, it follows that $\alpha'_1(y)=\sigma(y)$ for each $y\in\var(t,\phi)\cap\var(t',\phi')$.
Let $\alpha'$ be a valuation with $\alpha'(y)=\alpha'_1(y)$ for each $y\not\in\var(t,\phi)\cap\var(t',\phi')$. We obviously have $M^\Sigma,\alpha'\models \phi$ and  $\alpha'(\sigma(t'))=\alpha'_1(t')=\alpha'_1(t)=\alpha_1(t)=\alpha(t)$. Since $\alpha$ is arbitrary, we obtain $\tsem{\sigma(\ct{t}{\phi})}\subseteq \tsem{\sigma(\ct{t'}{\phi'})}$.
\\
\textit{Direct implication} ($\Rightarrow$). Assume that $\tsem{\ct{t}{\phi}}\subseteq_{\it shared} \tsem{\ct{t'}{\phi'}}$. We have to prove that
$M^\Sigma\models \phi\limplies (\exists \widetilde{x})(t \qeq t'\land \phi')$. 
Let $\alpha$ be a valuation such that $M^\Sigma,\alpha\models\sigma(\phi)$. Consider $\sigma:\var(t,\phi)\cap\var(t',\phi')\to M^\Sigma$ given by $\sigma(y)=\alpha(y)$. Since
$\tsem{\sigma(\ct{t}{\phi})}\subseteq \tsem{\sigma(\ct{t'}{\phi'})}$ it follows that there is $\alpha'$ such that $\alpha(\sigma(t))=\alpha'(\sigma(t'))$ and $M^\Sigma\!,\alpha'\models\sigma(\phi')$. 
We may assume w.l.o.g. that $\alpha'(y)=\alpha(y)$ for $y\not\in\var(t',\phi')$.
Let $\alpha'_1$ be the valuation such that $\alpha'_1(y)=\textrm{if }y\in\var(t,\phi)\cap\var(t',\phi')\textrm{ then }\allowbreak\sigma(y)\textrm{ else }\alpha'(y)$. 
We obviously have $\alpha'_1(y)=\alpha(y)$ for each $y\not\in\widetilde{x}$, $M^\Sigma,\alpha'_1\models\phi'$ (since $M^\Sigma,\alpha'\models\sigma(\phi')$),  and $\alpha'_1(t)=\alpha'(\sigma(t))=\alpha'(\sigma(t'))=\alpha'_1(t')$. 
Hence  $M^\Sigma,\alpha'_1\models (t \qeq t'\land \phi')$, which implies $M^\Sigma,\alpha\models (\exists \widetilde{x})(t \qeq t'\land \phi')$. Since $\alpha$ is arbitrar, it follows that $M^\Sigma\models \phi\limplies (\exists \widetilde{x})(t \qeq t'\land \phi')$.
\qed\end{proof}

\deltacommutes*
\begin{proof}
$\subseteq$.
Let $\gamma_2\in \tsem{\Delta_\R(\varphi)}$. There is a rule $\rrule{l}{r}{\phi_{lr}}$ in $\R$, a context $c[\cdot]$ and a valuation $\alpha$ such that $\gamma_2=\alpha(c[r])$ and $M^\Sigma,\alpha\models \phi'$ ($\phi'$ from the definition of $\Delta_{l,r,\phi_{\it lr}}$). Recall that  the variables in $\rrule{l}{r}{\phi_{lr}}$ are possibly renamed in order to be disjoint from those in $\varphi$.  $M^\Sigma,\alpha\models \phi'$ implies $M^\Sigma,\alpha\models \phi$ and $\alpha(t)=\alpha(c[l])$. 
It follows that $\gamma_1=\alpha(c[l])=\alpha(t)$ is in $\tsem{\varphi}$  and $\gamma_1\gtran{\R}\gamma_2$. Hence $\gamma_2\in\partial(\tsem{\varphi})$.
\\
$\supseteq$. Let $\gamma_2\in \partial(\tsem{\varphi})$. It follows that there is $\gamma_1\in\tsem{\varphi}$ s.t. $\gamma_1\gtran{} \gamma_2$ by the definition of $\partial$. The transition step $\gamma_1\gtran{} \gamma_2$ implies that there exist a rule $\rrule{l}{r}{\phi_{\it lr}}$ in $\R$, a context $c[\cdot]$, and a valuation $\alpha_1: X\to M^\Sigma$ such that $\gamma_1=\alpha_1(c[l])$, $\gamma_2=\alpha_1(c[r])$ and $M^\Sigma,\alpha_1\models \phi_{\it lr}$. 
Since $\gamma_1\in\tsem{\varphi}$, there exists $\alpha_2:X\to M^\Sigma$ such that $\gamma_1=\alpha_2(t)$ and $M^\Sigma,\alpha_2 \models \phi$.
We consider $\alpha$ such that $\alpha(x)=\alpha_1(x)$ for $x\in\var(\rrule{l}{r}{\phi_{\it lr}})$ and $\alpha(y)=\alpha_2(y)$ for $y\in\var(\varphi)$; this can be achieved by renaming the variables occurring in the rule.
The equality $\alpha(t)=\alpha(c[l])$  follows from $\alpha(t)=\gamma_1=\alpha(c[l])$. %
We obtain $M^\Sigma,\alpha \models c[l]  \qeq  t\land\phi_{\it lr}\land \phi$, which implies $\gamma_2\in \tsem{\Delta_\R(\varphi)}$. 
\qed\end{proof}

\begin{restatable}{corollary}{deltaequivalent}
\label{cor:der}
If $\rform{\varphi_1}{\varphi'}\equiv_\R\rform{\varphi_2}{\varphi'}$ then $\tsem{\Delta_\R(\varphi_1)}=\tsem{\Delta_\R(\varphi_2)}$.
\end{restatable}

\begin{proof}
We have $\tsem{\varphi_1}=\tsem{\varphi_2}$ by Remark~\ref{rem:eqrf}, which implies $\partial(\tsem{\varphi_1})=\partial(\tsem{\varphi_2})$.
Then we apply Theorem~\ref{th:der}.
\qed\end{proof}

\begin{remark}
\label{rem:eqder}

In this paper we assume more than Corollary~\ref{cor:der} claims,
namely that
$\rform{\varphi_1}{\varphi'}\equiv_\R\rform{\varphi_2}{\varphi'}$
implies \begin{enumerate}

\item[] for each $\varphi''_i\in \Delta_\R(\varphi_i)$ there is
  $\varphi''_{3-i}\in \Delta_\R(\varphi_{3-i})$ such that
  $\rform{\varphi''_i}{\varphi'}\equiv_\R\rform{\varphi_{3-i}}{\varphi'}$,
  $i=1,2$.

\end{enumerate}
In this way the symbolic execution given by the system {\sf
  DSTEP}$(\R)$ (see below) is preserved by the equivalence.
\end{remark}

\dstepsound*
\begin{proof}
\textit{Direct implication (soundness).}
Let $A$ be 
the set
\\
\indent{$\left\{\rform{\tsem{\sigma(\varphi)}}{\tsem{\sigma(\varphi')}}\,\middle|\, \sigma:\var(\varphi)\cap\var(\varphi')\to M^\Sigma, \rform{\varphi}{\varphi'}\in\allowbreak\nu\,\widehat{\sf DSTEP}(\R)\right\}$.}
\\
The conclusion of the theorem follows by showing that the set $A$ is backward closed w.r.t. $\widehat{\sf DVP}$, i.e. $A\subseteq \widehat{\sf DVP}(A)$.
Let $\rform{\varphi}{\varphi'}\in\nu\,\widehat{\sf DSTEP}(\R)$ and $\sigma:\var(\varphi)\cap\var(\varphi')\to M^\Sigma$. There is a proof tree $PT$ of  
$\rform{\varphi}{\varphi'}$ under ${\sf DSTEP}(\R)$, where $\rform{\varphi}{\varphi'}$ is the root. We proceed by case analysis on the {\sf DSTEP} rule applied to the root $\rform{\varphi}{\varphi'}\eqbydef\rform{\ct{t}{\phi}}{\ct{t'}{\phi'}}$. 
\\
1.[{\sf axiom}]. Then $M^\Sigma\models\phi\iff\False$ and $\rform{\tsem{\sigma(\varphi)}}{\tsem{\sigma(\varphi')}}\in \widehat{\sf DVP}(A)$ by the first rule of {\sf DVP}.
\\
2. [{\sf subs}]. The root has one child $\rform{\ct{t''}{\phi''\land \neg \phi'''}}{\ct{t'}{\phi'}}$, where 
$\rform{\varphi}{\varphi'}\equiv\rform{\ct{t''}{\phi''}}{\ct{t'}{\phi'}}$ (recall that the equivalence rule is implicitly applied)
and $M^\Sigma\models\phi'''\iff(\exists X)(t'' \qeq t'\land\phi')$, $X=\var(t',\phi')\setminus\var(t'',\phi'')$. 
The side condition of the rule ensures that
$M^\Sigma\models\neg(\phi'''\iff \False)$.%
We have $\tsem{\sigma(\ct{t''}{\phi'''})}\subseteq\tsem{\sigma(\ct{t'}{\phi'})}$ by Proposition~\ref{prop:subs}, which implies the equality
$\tsem{\sigma(\ct{t''}{\phi''{\land}\phi'''})}=\allowbreak\tsem{\ct{t''}{\phi''}}\cap\tsem{\sigma(\ct{t'}{\phi'})}$.
Moreover, from
\\%
\indent{$\tsem{\sigma(\ct{t''}{\phi''})}=\tsem{\sigma(\ct{t''}{\phi''\land\phi'''})}\uplus\tsem{\sigma(\ct{t''}{\phi''\land\neg\phi'''})}$}
\\%
we obtain
\\%
\indent{$\tsem{\sigma(\ct{t''}{\phi''\land\neg\phi'''})}=\tsem{\sigma(\ct{t''}{\phi''})}\setminus\tsem{\sigma(\ct{t''}{\phi''\land\phi'''})}$.}
\\%
We obviously have $\rform{\ct{t''}{\phi''\land \neg \phi'''}}{\ct{t'}{\phi'}}\in\nu\,\widehat{\sf DSTEP}(\R)$ (as a child of the proof tree root) and hence we get $\rform{\tsem{\sigma(\ct{t''}{\phi''\land\neg\phi'''})}}{\tsem{\sigma(\ct{t'}{\phi'})}}\in A$.
Moreover, $\tsem{\sigma(\ct{t}{\phi})}=\tsem{\sigma(\ct{t''}{\phi''})}$ by Remark~\ref{rem:eqrf}.
We distinguish two subcases:
\\
2.1. $M^\Sigma\models \phi'''\equiv\True$. Then $\tsem{\sigma(\ct{t''}{\phi''})}=\tsem{\sigma(\ct{t''}{\phi''\land\phi'''})}$, which
implies $\tsem{\sigma(\ct{t}{\phi})}\subseteq\allowbreak\tsem{\sigma(\ct{t'}{\phi'})}$. We obtain $\rform{\tsem{\sigma(\varphi)}}{\tsem{\sigma(\varphi')}}\in \widehat{\sf DVP}(A)$ by the first rule of {\sf DVP}.
\\
2.2. $M^\Sigma\models \neg(\phi'''\equiv\True)$. The children of $\rform{\ct{t''}{\phi''\land \neg \phi'''}}{\ct{t'}{\phi'}}$ in $\it PT$ are given by the rule [{\sf der}$^\forall$]. 
We proceed as in the case 3. and we get that $\tsem{\sigma(\ct{t''}{\phi''\land \neg \phi'''})}$ is runnable.
It follows $\rform{\tsem{\sigma(\ct{t}{\phi})}}{\tsem{\sigma(\ct{t'}{\phi'})}}\in \widehat{\sf DVP}(A)$ by the rule  $\semrule{Step}$ of {\sf DVP}.%
\\
3. [${\sf der}^\forall$]. The root has a set of children $\rform{\{\ct{t^j}{\phi^j}}{\varphi'}\mid j\in J\}$, where $J= \{1,\ldots,n\}$, $M^\Sigma\models \phi^j\iff(\phi''\land \phi'_j)$ with $\phi'_j$ depending only on $t''$ and the applied rule, and $\rform{\varphi}{\varphi'}\equiv\rform{\ct{t''}{\phi''}}{\ct{t'}{\phi'}}$.
We have:
\begin{align*}
&\rform{\ct{t^j}{\phi^j}}{\varphi'}\in \nu\,\widehat{\sf DSTEP}(\R) &&\textrm{implies by Prop.~\ref{pro:eqrform2}}\\
&\rform{\ct{z}{z \qeq t^j\land\phi^j}}{\varphi'}\in \nu\,\widehat{\sf DSTEP}(\R)&&\textrm{implies by Prop.~\ref{prop:disj}}\displaybreak[0]\\ 
&\rform{\ct{z}{\bigvee_{j\in J}(z \qeq t^j\land\phi^j)}}{\varphi'}\in \nu\,\widehat{\sf DSTEP}(\R)&&\textrm{implies by def. of }A\displaybreak[0]\\
&\rform{\tsem{\sigma(\ct{z}{\bigvee_{j\in J}(z \qeq t^j\land\phi^j)})}}{\tsem{\sigma(\varphi')}}\in A
\end{align*}
Since $\bigcup_{j\in J}\tsem{\sigma(\ct{t^j}{\phi^j})}=\tsem{\sigma(\ct{z}{\bigvee_{j\in J}(z \qeq t^j\land\phi^j)})}$ by Remark~\ref{rem:eqrf}, it follows that
$\bigcup_{j\in J}\rform{\tsem{\sigma(\ct{t^j}{\phi^j})}}{\tsem{\sigma(\varphi')}}\in A$.

We have $\rform{\tsem{\sigma(\ct{t^j}{\phi^j})}}{\tsem{\sigma(\varphi')}}\in A$ since $\rform{\ct{t^j}{\phi^j}}{\varphi'}\in \nu\,\widehat{\sf DSTEP}(\R)$, for each $j\in J$, and hence
$\rform{\bigcup_{j\in J}\tsem{\sigma(\ct{t^j}{\phi^j})}}{\tsem{\sigma(\varphi')}}\in A$ by the definition of $A$.
We have $\bigcup_{j\in J}\tsem{\ct{t^j}{\phi^j}}=\tsem{\Delta_\R(\ct{t''}{\phi''})}=\partial(\tsem{\ct{t''}{\phi''}}=\partial(\tsem{\ct{t}{\phi}})$ %
by Theorem~\ref{th:der} and Remark~\ref{rem:eqrf}, and hence the equality $\bigcup_{j\in J}\tsem{\sigma(\ct{t^j}{\phi^j})}=\partial(\tsem{\sigma(\ct{t}{\phi})})$. 
The side-condition of the inference rule implies $\tsem{\ct{t}{\phi}}=\tsem{\ct{t''}{\phi''}}$ runnable, and hence $\tsem{\sigma(\ct{t}{\phi})}$ runnable.
Since the side-condition of the inference rule implies $\tsem{\ct{t_l}{\phi_l}}\setminus\tsem{\ct{t_r}{\phi_r}}=\tsem{\ct{t_l}{\phi_l}}$, It follows that
$\rform{\tsem{\sigma(\ct{t}{\phi})}}{\tsem{\sigma(\ct{t'}{\phi'})}}\in \widehat{\sf DVP}(A)$ by the rule $\semrule{Step}$ of {\sf DVP}.

\noindent
\textit{Reverse implication (completeness).} Assume that
$\R\demModels\rform{\varphi}{\varphi'}$, i.e.,
$(M^\Sigma,\gtran{\R}) \demModels
\rform{\tsem{\sigma(\varphi)}}{\tsem{\sigma(\varphi')}}$ for all
$\sigma \in \SHSH$, where
$\SHSH = \{ \sigma \mid \sigma : \var(\varphi) \cap \var(\varphi') \to
M^\Sigma \}$. We prove by coinduction that there is a proof of
$\rform{\varphi}{\varphi'}$ under ${\sf DSTEP}(\R)$. We distinguish
two cases:
\begin{enumerate}
\item if there exists $\sigma \in \SHSH$ such that
  $\tsem{\sigma(\varphi)} \cap \tsem{\sigma(\varphi')} \not=
  \emptyset$, then we start our proof tree by a {\sf [subs]} node:

  \[\dfrac{\ldots}{[{\sf subs}]~\dfrac{\rform{\varphi \land \lnot \phi}{\varphi'}}{\rform{\varphi}{\varphi'},}}\]

  where $\phi = \exists \widetilde{x}.t_l = t_r \land \phi_r$ is the
  constraint in Rule~{\sf [subs]},
  Figure~\ref{fig:symbolic-execution}, under the assumption that
  $\varphi = \ct{t_l}{\phi_l}$ and $\varphi' = \ct{t_r}{\phi_r}$. By
  $\varphi \land \lnot \phi$ we denoted the constrained term
  $\ct{t_l}{\phi_l \land \lnot \phi}$. As $\tsem{\sigma(\varphi)} \cap
  \tsem{\sigma(\varphi')} \not= \emptyset$ for some $\sigma$, it
  follows that $\phi$ is satisfiable and therefore {\sf [subs]} can be
  applied. We have that $R \demModels \rform{\varphi \land \lnot
    \phi}{\varphi'}$ and also that, for any $\sigma \in \SHSH$,
    $\tsem{\sigma(\varphi\land \phi')} \cap
  \tsem{\sigma(\varphi')} = \emptyset$
(by the definition of
  $\phi$). Therefore, we continue to build the proof tree of
  $\rform{\varphi \land \lnot \phi}{\varphi'}$ coinductively (directly
  going into the second case, with $\varphi \land \lnot \phi$ playing
  the role of $\varphi$, and $\varphi'$ the role of $\varphi'$).

\item if for all $\sigma \in \SHSH$, $\tsem{\sigma(\varphi)} \cap
  \tsem{\sigma(\varphi')} = \emptyset$, we distinguish two more cases:

  \begin{enumerate}

  \item if for all $\sigma \in \SHSH$, $\tsem{\sigma(\varphi)} =
    \emptyset$, then we construct a proof tree of
    $\rform{\varphi}{\varphi'}$ as follows:

    \[{\sf [axiom]}~\dfrac{\ }{\rform{\varphi}{\varphi'}.}\]

    \item if there exists $\sigma \in \SHSH$ such that
      $\tsem{\sigma(\varphi)} \not= \emptyset$, let $\SHSH_1 = \{
      \sigma \in \SHSH \mid \tsem{\sigma(\varphi)} \not= \emptyset\}$
      and $\SHSH_2 = \{ \sigma \in \SHSH \mid \tsem{\sigma(\varphi)}
      = \emptyset\}$. We have that $\SHSH_1 \not= \emptyset$.

      We have that for all $\sigma \in \SHSH_1$: \begin{enumerate*}
        \item[(A)] $\tsem{\sigma(\varphi)} \cap
          \tsem{\sigma(\varphi')} = \emptyset$, \item[(B)]
          $\tsem{\sigma(\varphi)} \not= \emptyset$, 
          and \item[(C)] $R
          \demModels
          \rform{\sigma(\varphi)}{\sigma(\varphi')}$. \end{enumerate*}
      Therefore, rule $\semrule{Step}$ must have been applied with $P
      = \tsem{\sigma(\varphi)}$ and $Q = \tsem{\sigma(\varphi')}$ to
      justify $R \demModels \rform{\sigma(\varphi)}{\sigma(\varphi')}$
      (for all $\sigma \in \SHSH_1$). But $P \cap Q = \emptyset$ and
      therefore $P \setminus Q = P$. Which means $R \demModels
      \rform{\partial(P)}{Q}$ and $P$ is runnable, that is $R
      \demModels
      \rform{\partial(\tsem{\sigma(\varphi)})}{\tsem{\sigma(\varphi')}}$
      and $\tsem{\sigma(\varphi)}$ is runnable (for all $\sigma \in
      \SHSH_1$).

      We have that $\tsem{\Delta_R(\varphi)} = \allowbreak
      \partial(\tsem{\varphi}) = %
      \bigcup_{\sigma \in \SHSH_1}(\partial(\tsem{\sigma(\varphi)}))
      \not= \emptyset$ (since $\tsem{\sigma(\varphi)}$ is not empty
      for $\sigma \in \SHSH_1$). Therefore, $\varphi$ is
      $R$-derivable.

      We show that we can apply the $[{\sf der}^\forall]$ rule to
      $\rform{\varphi}{\varphi'}$. Let $\varphi =
      \ct{t_l}{\phi_l}$. Indeed, as for all $\sigma \in \SHSH_1$,
      $\tsem{\sigma(\varphi)}$ is runnable, it follows that for all
      $\rho \in \sem{\phi_l}$, there is a rewrite rule
      $\rrule{l}{r}{\phi_{lr}}$, a ground context $c[\cdot]$ and
      $\rho' : \var(l, r, \phi_{lr})$ such that $\rho(\sigma(\varphi))
      = \rho'(c[l])$ and $\rho(\phi_{lr}) = \True$.

      This means that $\phi_l \limplies \bigvee_{j \in \{1, \ldots,
        n\}}\exists{\widetilde{y}^j}.\phi^j$ is valid in rule $[{\sf
          der}^\forall]$ in Figure~\ref{fig:symbolic-execution} and
      therefore it can be applied:

      \[[{\sf der}^\forall]~\dfrac{\dfrac{\ldots}{\rform{\ct{t^j}{\phi^j}}{\varphi'} \mbox{, $j \in \{ 1, \ldots, n \}$}}}{\rform{\varphi}{\varphi'}}.\]

      Next we show that the coinduction hypothesis can be applied on
      all hypotheses $\rform{\ct{t^j}{\phi^j}}{\varphi'}$. Indeed, it
      is sufficient to show that $R \demModels
      \rform{\ct{t^j}{\phi^j}}{\varphi'}$. First, notice that
      $\var(\varphi) \cap \var(\varphi') = \var(\Delta_R(\varphi))
      \cap \var(\varphi')$ (taking the derivative of
      $\varphi$ preserves the common variables with $\varphi'$).

      By Proposition~\ref{prop:redrp}, it is sufficient to show that
      for all $\sigma \in \SHSH$,
      $R \demModels
      \rform{\tsem{\sigma(\Delta_R(\varphi))}}{\tsem{\sigma(\varphi')}}$.
      But
      $\tsem{\sigma(\Delta_R(\varphi))} =
      \tsem{\Delta_R(\sigma(\varphi))}
      = \partial(\tsem{\sigma(\varphi)})$, and
      $R \demModels
      \rform{\partial(\tsem{\sigma(\varphi)})}{\tsem{\sigma(\varphi')}}$
      is already known to hold.
  \end{enumerate}
\end{enumerate}

We have shown that in whenever $R \demModels
\rform{\varphi}{\varphi'}$, we can build a proof tree of
$\rform{\varphi}{\varphi'}$ under ${\sf DSTEP}(\R)$, which concludes
the proof of the completeness of ${\sf DSTEP}(\R)$.
\qed
\end{proof}

\circularityprinciple*
\begin{proof}
We first introduce some notations.
Let  $\sqsubset$ be the partial order over proof trees of {\sf DCC} defined as follows: ${\it PT}_1\sqsubset {\it PT}_2$ iff ${\it circOut}({\it PT}_1)$ is a subtree of ${\it circOut}({\it PT}_2)$, where ${\it circOut}({\it PT})$ is the tree obtained from $\it PT$ by removing all subtrees having a {\sf [circ]} root. 
If ${\it circHeight}({\it PT})$ denote the length of the shortest path from the root to a {\sf circ}-node in the proof tree $\it PT$ under {\sf DCC}, then 
${\it PT}_1\sqsubset {\it PT}_2$ implies ${\it circHeight}({\it PT}_1)\le{\it circHeight}({\it PT}_2)$. The main idea of the proof is to transform a guarded proof tree $\it PT$ under {\sf DCC}$(\R,G)$ for a $\rform{\varphi}{\varphi}'\in G$ into a proof tree ${\it PT}'$ under {\sf DSTEP}$(\R)$ for the same formula, where ${\it PT}'$ is the lub $\bigsqcup_{i\ge 0}{\it PT}_i$ of a chain ${\it PT}={\it PT}_0\sqsubset {\it PT}_1\sqsubset {\it PT}_2\sqsubset \cdots$ with the property that ${\it circHeight}({\it PT}_i)<{\it circHeight}({\it PT}_{i+1})$ (this ensures that the limit $\bigsqcup_{i\ge 0}{\it PT}_i$ has no {\sf [circ]}-nodes).
We show how ${\it PT}_{i+1}$ is obtained from ${\it PT}_i$. Let $\rform{\ct{t_i}{\phi_i}}{\varphi'}$ a {\sf circ} node that gives ${\it circHeight}({\it PT}_i)$, i.e. its children are 
$\rform{\ct{t'_c}{\phi'_c\land\phi_i\land\phi''_i)}}{\varphi'}$ and
$\rform{\ct{t_i}{\phi_i\land\neg\phi''_i}}{\varphi'}$,
where 
$M^\Sigma\models\phi''_i\iff (\exists\var(t_c,\phi_c))(t_i \qeq t_c\land\phi_c)$ and 
$\rform{\ct{t_c}{\phi_c}}{\ct{t'_c}{\phi'_c}}\in G$.

$(\R,G)\demEntails G$ implies $(\R,G)\demEntails\rform{\ct{t_c}{\phi_c}}{\ct{t'_c}{\phi'_c}}$ and hence we obtain $(\R,G)\demEntails\rform{\ct{t_c}{\phi_c\land\phi_i\land\phi''_i}}{\ct{t'_c}{\phi'_c}}$ by Proposition~\ref{prop:dcc-impl}.

Let ${\it PT}_c$ be a guarded proof tree for $\rform{\ct{t_c}{\phi_c\land\phi_i\land\phi''_i}}{\ct{t'_c}{\phi'_c}}$ under {\sf DCC}$(\R,G)$
and we want to transform it into a proof tree ${\it PT}'_c$ under the same proof system for $\rform{\ct{t_c}{\phi_c\land\phi_i\land\phi''_i}}{\varphi'}$. 
We may replace all the rhs $\ct{t'_c}{\phi'_c}$ by $\varphi'$ in ${\it PT}_c$ (in the sense that they are valid instances of the {\sf DCC}$(\R,G)$ rules) excepting the nodes that are instances of the inference rule [{\sf subs}], because this rule involves the right-hand side of the reachability formula.
Let $\rform{\ct{t_j}{\phi_j}}{\ct{t'_c}{\phi'_c}}$ be a [{\sf subs}]-node in ${\it PT}_c$, i.e. its child is (or equivalent to)
$\rform{\ct{t_j}{\phi_j\land\neg\phi''_j}}{\ct{t'_c}{\phi'_c}}$, where $M^\Sigma\models \phi''_j\iff (\exists X)(t_j \qeq t'_c\land\phi'_c)$
and $X=\var(t'_c,\phi'_c)\setminus\var(t_j,\phi_j)$.
Our intention is to transform this node into a [{\sf disj}]-node 
$\rform{\ct{t_j}{\phi_j}}{\varphi'}$ with the children $\rform{\ct{t_j}{\phi_j\land\phi''_j}}{\varphi'}$ and $\rform{\ct{t_j}{\phi_j\land\neg\phi''_j}}{\varphi'}$. 
In order to obtain ${\it PT}'_c$ a valid proof tree, we have to add to it, as a subtree of the new node, a proof tree for $\rform{\ct{t_j}{\phi_j\land\phi''_j}}{\varphi'}$ or for a formula equivalent to it. 
We know that $\rform{\ct{t'_c}{\phi'_c\land\phi_i\land\phi''_i}}{\varphi'}$ is a node in ${\it PT}_i$ and hence there is a proof tree for it.
We show that $\tsem{\ct{t_j}{\phi_j\land\phi''_j}}\subseteq\tsem{\ct{t'_c}{\phi'_c\land\phi_i\land\phi''_i}}$.
Let $\rho$ be in $\sem{\phi_j\land(\exists X)(t_j \qeq t'_c\land\phi'_c)}$. Note that $X=\var(\varphi'_c)\setminus\var(\varphi_j)$. There is $\rho'$ such that $\rho'(y)=\rho(y)$ for all $y\not\in X$ and $M^\Sigma,\rho'\models t_j \qeq t'_c\land\phi'_c$. 
Since $\rho$ and $\rho'$ coincide on $\var(\phi_j)$ we obtain $M^\Sigma,\rho'\models \phi_j\land t_j \qeq t'_c\land\phi'_c$ and hence
$M^\Sigma,\rho'\models \phi'_c\land(\exists X')(t_j \qeq t'_c\land\phi_j)$. Now the proof of the inclusion is finished.

Since $\var(\ct{t'_c}{\phi'_c\land\phi_i\land\phi''_i})\cap\var(\varphi')=\var(\phi_i\land\phi''_i)\cap\var(\varphi')$ and
$\rform{\ct{t_j}{\phi_j}}{\ct{t'_c}{\phi'_c}}$ is a node in the proof tree of $\rform{\ct{t_c}{\phi_c{\land}\phi_i{\land}\phi''_i}}{\ct{t'_c}{\phi'_c}}$ under {\sf DCC}$(\R,G)$, it follows that $\var(\ct{t_j}{\phi_j})\cap\var(\varphi')=\var(\phi_i\land\phi''_i)\cap\var(\varphi')$.
It follows that $\rform{\ct{t_j}{\phi_j\land\phi''_j}}{\varphi'}$ is equivalent to $\rform{\ct{t'_c}{t_j \qeq t'_c\land \phi_j\land\phi''_j\land \phi'_c\land\phi_i\land\phi''_i}}{\varphi'}$ and the later one has a proof tree under {\sf DCC}$(\R,G)$ by 
by Proposition~\ref{prop:dcc-impl}. This proof tree is added as the subtree of $\rform{\ct{t_j}{\phi_j\land\phi''_j}}{\varphi'}$.
Now the transformation of $\it PT_c$ into $\it PT'_c$ is completely described.

The proof tree ${\it PT}_{i+1}$ is the result of processing all {\sf [circ]} nodes that give ${\it circHeight}({\it PT}_i)$. The relation ${\it PT}_i\sqsubset{\it PT}_{i+1}$ is given by the fact that  ${\it PT}_i$ is guarded.
Moreover, we have ${\it circHeight}({\it PT}_{i})<{\it circHeight}({\it PT}_{i+1})$.
\qed\end{proof}

\end{document}